\documentclass[journal]{IEEEtran}
\usepackage[final]{graphicx}
\usepackage{subfigure}
\usepackage{float}
\usepackage{amsmath}
\usepackage{cases}
\usepackage{color}
\usepackage{colortbl}
\usepackage{algorithm}
\usepackage{algpseudocode}
\usepackage{amsmath}
\usepackage{amssymb}
\usepackage{booktabs}
\usepackage{setspace}
\usepackage{array}
\usepackage{rotating}
\usepackage[normalem]{ulem}
\usepackage{bbding}

\renewcommand{\normalsize}{\fontsize{10}{12}\selectfont}

\newtheorem{lemma}{\noindent \bf Lemma}
\newtheorem{theorem}{ \bf Theorem}
\newenvironment{proof}{{ \noindent \it Proof.}}{\hfill $\blacksquare$}

\makeatletter
\renewcommand{\maketag@@@}[1]{\hbox{\m@th\normalsize\normalfont#1}}%
\makeatother
\usepackage{stfloats}
\usepackage{cite}
\usepackage{makecell}
\usepackage{multirow}
\usepackage{hyperref}

\ifCLASSINFOpdf
\else
\fi

\usepackage{caption}
\usepackage{mathtools}
\UseRawInputEncoding

\begin{document}
\title{Cooperative Sensing in Cell-free Massive MIMO ISAC Systems: Performance Optimization and Signal Processing}
\author{Haotian Liu,~\IEEEmembership{Student Member,~IEEE,}
Zhiqing Wei,~\IEEEmembership{Member,~IEEE,}
Luyang Sun,~\IEEEmembership{Student Member,~IEEE,}\\
Ruizhong Xu,~\IEEEmembership{Student Member,~IEEE,}
Yixin Zhang,~\IEEEmembership{Member,~IEEE,}
Zhiyong Feng,~\IEEEmembership{Senior Member,~IEEE}

\thanks{Haotian Liu, Zhiqing Wei, Luyang Sun, Ruizhong Xu, Yixin Zhang, and Zhiyong Feng 
are with the Beijing University of Posts and Telecommunications, 
Beijing 100876, China (e-mails: \{haotian\_liu; weizhiqing; sly1105; xuruizhong; yixin.zhang; fengzy\}@bupt.edu.cn). 
\textit{Corresponding authors: Zhiqing Wei}
}
}

\maketitle

\begin{abstract}
Emerging applications such as low-altitude economy and intelligent transportation
hold the promise of significant economic and social benefits,
requiring the support of technologies that integrate robust communication with precise sensing. 
Integrated sensing and communication (ISAC),
as a technology enabled seamless connection between communication and sensing,
is regarded a core enabling technology for these applications.
However, the accuracy of single-node sensing in ISAC system is limited,
prompting the emergence of multi-node cooperative sensing.
In multi-node cooperative sensing,
the synchronization error limits the sensing accuracy,
which can be mitigated by the architecture of cell-free massive multi-input multi-output (CF-mMIMO),
since the multiple nodes are interconnected via optical fibers with high synchronization accuracy.
However, the multi-node cooperative sensing in CF-mMIMO ISAC systems faces the following challenges:
1) The joint optimization of placement and resource allocation of
distributed access points (APs) to improve the sensing performance
in multi-target detection scenario is difficult;
2) The fusion of the sensing information from distributed APs 
with multi-view discrepancies is difficult.
To address these challenges,
this paper proposes a joint placement and antenna resource optimization scheme
for distributed APs to minimize the sensing Cramér-Rao bound for targets' parameters within the area of interest.
Then, a symbol-level fusion-based multi-dynamic target sensing (SL-MDTS) scheme is provided,
effectively fusing sensing information from multiple APs.
The simulation results validate the effectiveness of the joint optimization scheme
and the superiority of the SL-MDTS scheme.
Compared to state-of-the-art grid-based symbol-level sensing information fusion schemes, 
the proposed SL-MDTS scheme improves the accuracy of localization and velocity estimation
by 44\% and 41.4\%, respectively.
\end{abstract}

\IEEEpeerreviewmaketitle

\section{Introduction}
\subsection{Background and Motivations}
Emerging applications such as the low-altitude economy (LAE), 
autonomous vehicle networks, 
and smart cities are expected to hold the promise of significant economic and social benefits, 
which require the support of technologies
that integrate robust communication with precise sensing~\cite{feng2025}.
Integrated sensing and communication (ISAC),
as a key technology of sixth-generation mobile networks (6G),
combines communication and sensing functions by reusing spectrum and hardware resources~\cite{wei2024integrated}.
Compared to the traditional system with separate communication and sensing functions,
ISAC enables a seamless connection between precise sensing and
robust communication within a unified system
and is expected to be a core enabling technology
for the emerging applications mentioned above~\cite{wei2024integrated}.
However, the limitations of single-node sensing in ISAC systems caused by occlusion and limited
viewpoints make it difficult to meet high-precision sensing demands
and adapt to complex application environments.
Consequently, multi-node cooperative sensing has garnered
significant attention from researchers~\cite{zeng2023integrated,wei2024integrated,meng2025, wei_symbol}.
In multi-node cooperative sensing,
the synchronization accuracy between multiple nodes limits the sensing performance.
The cell-free massive multiple input multiple output (CF-mMIMO) architecture
enables multiple distributed access points (APs) connected via optical fibers
to perform cooperative sensing with high synchronization accuracy,
making it a key enabler in cooperative ISAC~\cite{Du_cell,Shi_sensing,Sakhnini_cell}.

Currently, extensive work has concentrated on
interference mitigation~\cite{Fan_inter},
beamforming design~\cite{Mao_beam,Demirhan_beam},
resource allocation~\cite{zeng2023integrated,Behdad_power},
and target detection~\cite{Sakhnini_cell,elfiatoure_detection},
establishing a robust theoretical foundation for cooperative sensing of the CF-mMIMO ISAC network.
However, there are still some challenges:
1) \textit{Joint parameter optimization}:
The joint optimization of placement and antenna resource allocation of
distributed APs to maximize the sensing performance 
in multi-target scenarios is difficult;
2) \textit{Sensing information fusion}:
The fusion of the sensing information
from distributed APs with multi-view discrepancies is difficult. 

\begin{table*}[!htbp]
\centering
\caption{Comparison of this work with existing works
on CF-mMIMO ISAC/distributed MIMO radar/multi-BS ISAC
cooperative sensing}
\label{tab1}
\resizebox{\textwidth}{!}{%
\renewcommand{\arraystretch}{1}
\begin{tabular}{|c|cccc|c|ccccc|}
\hline
\multirow[b]{2}{*}{\textbf{\begin{tabular}[c]{@{}c@{}}Existing \\ works\end{tabular}}} & \multicolumn{4}{c|}{\textbf{AP Deployment and Antenna Allocation}}                    & \multirow[b]{2}{*}{\textbf{\begin{tabular}[c]{@{}c@{}}Existing \\ works\end{tabular}}} & \multicolumn{5}{c|}{\textbf{Cooperative Sensing Signal Processing}}               \\ \cline{2-5} \cline{7-11} 
& \multicolumn{1}{c|}{\begin{tabular}[c]{@{}c@{}}Multiple\\ targets\end{tabular}} & \multicolumn{1}{c|}{\begin{tabular}[c]{@{}c@{}}MIMO \\ channel\end{tabular}} & \multicolumn{1}{c|}{\begin{tabular}[c]{@{}c@{}}Dynamic\\ target\end{tabular}} & \begin{tabular}[c]{@{}c@{}}CRB\\ criterion\end{tabular} &                                 & \multicolumn{1}{c|}{\begin{tabular}[c]{@{}c@{}}Multiple\\ targets\end{tabular}} & \multicolumn{1}{c|}{\begin{tabular}[c]{@{}c@{}}MIMO\\ channel\end{tabular}} & \multicolumn{1}{c|}{\begin{tabular}[c]{@{}c@{}}Symbol-level\\ fusion\end{tabular}} & \multicolumn{1}{c|}{\begin{tabular}[c]{@{}c@{}}Dynamic\\ target\end{tabular}} & \begin{tabular}[c]{@{}c@{}}Continuous\\ estimation\end{tabular} \\ \hline
\cite{Rui_radar}   & \multicolumn{1}{c|}{}   & \multicolumn{1}{c|}{}                     & \multicolumn{1}{c|}{}                     & \Checkmark                                & \cite{Shi_sensing,Zhang_sensing}          & \multicolumn{1}{c|}{\Checkmark }          & \multicolumn{1}{c|}{}           & 
\multicolumn{1}{l|}{}                     & \multicolumn{1}{l|}{}                     & \Checkmark  \\ \hline
\cite{Nguyen_radar} & \multicolumn{1}{c|}{}  & \multicolumn{1}{c|}{}                     & \multicolumn{1}{c|}{}                     & \Checkmark                                & \cite{Liu_sensing}                        & \multicolumn{1}{c|}{\Checkmark}          
 & \multicolumn{1}{c|}{}          & 
 \multicolumn{1}{l|}{}                    & \multicolumn{1}{l|}{}                    & \Checkmark   \\ \hline
\cite{moreno_radar} & \multicolumn{1}{c|}{\Checkmark}                                 & \multicolumn{1}{c|}{}                     & \multicolumn{1}{c|}{}                     & \Checkmark                                & \cite{wei_symbol}                         & \multicolumn{1}{l|}{}                     & \multicolumn{1}{l|}{}                     & \multicolumn{1}{c|}{\Checkmark}           & \multicolumn{1}{c|}{\Checkmark}           & \multicolumn{1}{l|}{} \\ \hline
\cite{liang_radar}                           & \multicolumn{1}{c|}{\Checkmark}           & \multicolumn{1}{c|}{}                     & \multicolumn{1}{c|}{}                     & \Checkmark                                & \cite{wei2024integrated,Lu_sensing}       & \multicolumn{1}{l|}{}                     & \multicolumn{1}{c|}{\Checkmark}           & \multicolumn{1}{c|}{\Checkmark}           & \multicolumn{1}{c|}{\Checkmark}           & \multicolumn{1}{l|}{} \\ \hline
\textbf{This work} & \multicolumn{1}{c|}{\Checkmark} & \multicolumn{1}{c|}{\Checkmark} & \multicolumn{1}{c|}{\Checkmark} & \Checkmark                 & \textbf{This work  } & \multicolumn{1}{c|}{\Checkmark} & \multicolumn{1}{c|}{\Checkmark} & \multicolumn{1}{c|}{\Checkmark} & \multicolumn{1}{c|}{\Checkmark} & \Checkmark \\ \hline
\end{tabular} }
\end{table*}

\subsection{Related Work}

Research on joint placement and antenna resource optimization for
APs and sensing information fusion
in the CF-mMIMO ISAC system is still in its early stages. 
However, related studies in distributed MIMO radar and multi-base station (BS)
cooperative sensing provide valuable insights.
Table \ref{tab1} compares this work with existing studies
on CF-mMIMO ISAC/distributed MIMO radar/multi-BS ISAC cooperative sensing.

In the AP deployment and antenna allocation field of CF-mMIMO ISAC systems,
substantial research has been conducted in areas such as
distributed MIMO radars with similar structural configurations.
Cramér-Rao Bound (CRB) is commonly used as a criterion in deployment optimization~\cite{Rui_radar,Nguyen_radar,moreno_radar,liang_radar}. 
Existing work has proposed various CRB-based schemes,
including optimizing receiver location for a single target by minimizing CRB~\cite{Rui_radar}, 
determining the optimal angular geometry by maximizing the determinant of the
Fisher information matrix (FIM) under single-target conditions~\cite{Nguyen_radar}, 
and extending these schemes to multi-target scenarios by optimizing the mean value of the FIM~\cite{moreno_radar}. 
Furthermore, the optimal placement of transmitters and receivers
for multiple static targets has been addressed~\cite{liang_radar}. 
However, these studies focus on single input single output channels, 
overlooking MIMO channels and excluding multi-dynamic target scenarios. 
Therefore, this paper adopts the CRB criterion and explores the joint optimization of
AP deployment and antenna allocation in the context of multi-dynamic targets and MIMO channels,
aiming to enhance sensing performance in the area of interest.

Cooperative sensing signal processing algorithms are rarely explored in CF-mMIMO ISAC systems, 
whereas extensive research exists in the area of multi-BS cooperative sensing.
Shi \textit{et al.} in~\cite{Shi_sensing} proposed
a circular localization method based on maximum likelihood estimation (MLE),
while Zhang \textit{ et al.} in~\cite{Zhang_sensing} introduced an
MLE-based ellipse localization method in a scenario with multiple targets.
Furthermore, Liu \textit{et al.} in~\cite{Liu_sensing} proposed a low-complexity multi-target matching algorithm.
These methods are treated as data-level fusion based target localization approaches, 
which have low sensing accuracy and are hard to meet the growing sensing demands of emerging applications, such as LAE applications.
To this end, Wei \textit{et al.} in~\cite{wei_symbol} introduced the concept of symbol-level fusion and
proposed a target sensing method for localization and absolute velocity estimation.
Subsequently, derivative algorithms based on symbol-level fusion are proposed, 
including three-dimensional (3D) sensing~\cite{Lu_sensing} 
and ellipse sensing~\cite{wei2024integrated}. 
However, these symbol-level fusion methods
rely on grid-based search,
which suffers from severe off-grid problems,
leading to insufficient sensing accuracy.
Therefore, this paper considers symbol-level fusion with
better fusion performance than data-level fusion to
meet the sensing requirements of emerging applications and
proposes a multi-dynamic target sensing method,
addressing the off-grid issue and significantly improving sensing accuracy.

\subsection{Our Contributions}
We consider cooperative sensing of the CF-mMIMO ISAC system
in a multi-dynamic target scenario,
where multiple APs cooperatively sense the target area.
A two-phase cooperative sensing framework is developed that
integrates joint optimization and cooperative sensing signal processing 
to improve the sensing capabilities of CF-mMIMO ISAC systems, 
satisfying the sensing requirements of emerging applications. 
The main contributions of this paper are summarized as follows.
\begin{itemize}
    \item \textbf{\textit{Two-phase cooperative sensing framework:}} 
    We first present a joint optimization scheme to minimize the sensing CRBs,
    providing system parameters for high-accuracy cooperative sensing, 
    including the placement and number of antennas for distributed APs. 
    Then, we present a cooperative sensing signal processing scheme to achieve high-accuracy sensing of multiple dynamic targets.
    \item \textbf{\textit{Joint optimization scheme:}} 
     A joint deployment and antenna resource optimization scheme
     for distributed APs is proposed to improve the theoretical sensing performance.
     Specifically, we derive the CRBs for localization and absolute velocity estimation under multi-AP cooperation. 
     Using the CRB criterion, we formulate a joint optimization problem, 
     solved via the alternating direction method of multipliers (ADMM), 
     the analytical method, and the truncated Newton method.
     The scheme improves the theoretical sensing performance and provides 
     the system parameters for high-accuracy cooperative sensing.
    \item \textbf{\textit{Cooperative sensing signal processing scheme:}} 
    We propose a symbol-level fusion-based multi-dynamic target sensing (SL-MDTS) scheme 
    for high-precision localization and absolute velocity estimation, 
    which comprises two stages: signal preprocessing and symbol-level fusion-based sensing. 
    In signal preprocessing stage, multi-target echo signals 
    are separated and associated, 
    transforming the multi-target sensing problem into parallel single-target sensing problems. 
    The symbol-level fusion-based sensing stage proposes a sensing method based on optimization algorithm and symbol-level fusion, 
    referred to as ``symbol-level fusion-based optimization (SFO) method'', 
    enabling efficient sensing information fusion and continuous estimations of target's location and absolute velocity.
    \item \textbf{\textit{Performance Verification:}} 
    The simulation results demonstrate the feasibility and superiority of the proposed two-phase cooperative sensing framework. 
    First, we confirm that the joint optimization scheme can improve the theoretical sensing performance. 
    Then, simulations validate that the proposed SL-MDTS scheme improves localization and velocity estimation accuracy
    by 44\% and 41.4\% respectively, compared to the state-of-the-art grid-based symbol-level sensing information fusion schemes. 
    Finally, the simulation results reveal that the joint optimization of system parameters can further improve the performance of cooperative sensing.
\end{itemize}

The remainder of this paper is structured as follows. 
Section~\ref{se2} outlines the system model, 
including signal model and performance metrics. 
In Section~\ref{se3}, we present a joint optimization scheme. 
Section~\ref{se4} proposes a cooperative sensing signal processing scheme. 
Section \ref{se5} details the simulation results, 
and Section \ref{se6} summarizes this paper.

\textit{Notations:} $\{\cdot\}$ typically represents a set of various index values. 
$\{A_i\}_{i=1}^I$ denotes the set of $I$ elements.
Vectors and matrices are written in bold letters and in capital bold letters, respectively. 
$\left[A_i\right]|_{i=1:L}$ denotes a vector consisting of the elements $A_1, A_2,\cdots,A_L$.
$\mathbb{C}$ and $\mathbb{R}$ denote the set of complex and real numbers, respectively. 
$\left[\cdot\right]^{\text{T}}$, $\left[\cdot\right]^{-1}$, $|\cdot|$, 
and $\langle\cdot\rangle$ stand for the transpose operator, the inverse operator, 
the absolute operator, and the inner product, respectively. 
$\text{tr}\left(\cdot\right)$ represents the trace of a matrix. 
$\circ $, $\odot$, and $\otimes$ denote the outer product, the Hadamard product, and the Khatri-Rao product, respectively. 
$\left\|\cdot\right\|_{2}$ is the Frobenius norm. 
A complex Gaussian random variable $\mathbf{u}$ with mean $\mu_u$ and variance $\sigma_u^2$ 
is denoted by $\mathbf{u} \sim \mathcal{CN}\left(\mu_u,\sigma_u^2\right)$.

\section{System Model}\label{se2}
As shown in Fig. \ref{fig1}, we consider a cooperative sensing scenario within a CF-mMIMO ISAC system, with $L$ APs cooperatively sensing the multiple targets.
The $l \in\{1,2,\cdots, L\}$-th AP, located at $\mathbf{c}_l=\left[x_l,y_l\right]^\text{T}$, is equipped full-duplex uniform linear arrays with antenna spacing $d_\text{r}$ and $N_\text{A}^l$ antennas, operating in time-division duplex for communication and sensing~\cite{elfiatoure_detection}.
Each AP transmits and receives ISAC signals for cooperative sensing, transmitting results to a central processing unit for fusion~\cite{wei_symbol,Mao_beam}. 
We consider $U$ targets in the sensing area, with the location and absolute velocity of the $u\in\{1,2,\cdots, U\}$-th target 
denoted as $\mathbf{t}_u=\left[x_u^\text{tr}, y_u^\text{tr}\right]^\text{T}$ and $\mathbf{v}_u=\left[v_u, \theta_u\right]^\text{T}$, respectively.
The timing information is transmitted to the APs via fronthaul links to synchronize them~\cite{Zhang_sensing}. 
Although this study focuses on a two-dimensional scenario, the approach extends to a three-dimensional scenario.
\begin{figure}
    \centering    \includegraphics[width=0.40\textwidth]{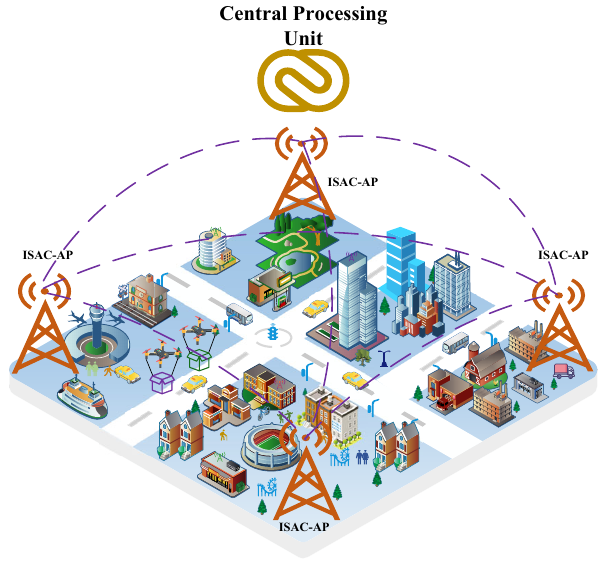}
    \caption{CF-mMIMO ISAC system cooperative sensing scenario}
    \label{fig1}
\end{figure}

\subsection{Sensing Signal Model}
Each AP occupies $N_\text{c}$ subcarriers and $M$ OFDM symbols for sensing, 
employing orthogonal waveform to avoid interference~\cite{chen2025}.
For the $l$-th AP, the ISAC echo signal in the $n\in\{1,2,\cdots,N_\text{c}\}$-th subcarrier during the $m\in\{1,2,\cdots,M\}$-th OFDM symbol period is expressed as
\begin{equation}\label{eq1}
  \mathbf{y}_{n,m}^l = \sqrt{P_\text{t}}\sum_{u=1}^{U}
\left[\begin{array}{l}
\alpha_l^u e^{j2\pi f_{\text{D},l}^u m T} e ^{-j2\pi n \Delta f \tau_l^u} \\ \times \mathbf{a}_\text{r}\left(\theta_l^u\right)
\chi_\text{t}^u d_{n,m}^l + \mathbf{z}_{n,m}^{l,u}
\end{array}\right],  
\end{equation}
where $\mathbf{y}_{n,m}^l \in\mathbb{C}^{N_\text{A}^l \times 1}$, 
and $\alpha_l^u = \sqrt{\frac{\lambda^2}{\left(4\pi\right)^3\left(r_l^u\right)^4}}\beta_S^u$ 
denotes the attenuation between the $l$-th AP and the $u$-th target with $\beta_S^u\sim \mathcal{CN}\left(0,\sigma_\beta^2\right)$ 
being the radar cross section~\cite{liu2024carrier}; 
$P_\text{t}$ is the transmit power, where each AP is assumed to have equal power; 
$\lambda = \frac{c}{f_\text{c}}$ is wavelength with $c$ and $f_\text{c}$ being the speed of light and carrier frequency, respectively; 
$\tau_l^u = \frac{2r_l^u}{c}$ represents the delay with $r_l^u = \|\mathbf{c}_l-\mathbf{t}_u\|_2$ being the distance between the $l$-th AP and the $u$-th target; $f_{\text{D},l}^u=\frac{-2f_\text{c}v_u\cos\left(\theta_l^u-\theta_u\right)}{c}$ is Doppler frequency shift, 
$\Delta f$ is subcarrier spacing, and $T$ is total duration of OFDM symbol; 
$\theta_l^u$ is the angle of arrive (AoA) between the $l$-th AP and the $u$-th target; 
$\chi_\text{t}^u = \mathbf{a}_\text{t}^\text{T}\left(\theta_l^u\right)\mathbf{w}_{\text{t}}^u$ is transmit beam gain with $\mathbf{w}_{\text{t}}^u\in\mathbb{C}^{N_\text{A}^l\times 1}$ being the transmit beamforming vector~\cite{Mao_beam}; 
$d_{n,m}^l$ is known pilot data and $\mathbf{z}_{n,m}^{l,u}\sim\mathcal{CN}\left(0,\sigma_l^2\right)$ is an additive white Gaussian noise (AWGN) vector; 
$\mathbf{a}_\text{r}\left(\cdot\right)$ and $\mathbf{a}_\text{t}\left(\cdot\right)$ are the transmit and receive steer vectors, expressed in (\ref{eq2}) and (\ref{eq3}), respectively. 
\begin{equation}\label{eq2}
\mathbf{a}_\text{r}\left(\cdot\right) = \left[e^{j2\pi p(\frac{d_\text{r}}{\lambda})\sin(\cdot)}\right]^\text{T}|_{p=1:N_\text{A}^l},
\end{equation}
\begin{equation}\label{eq3}
\mathbf{a}_\text{t}\left(\cdot\right) = \left[e^{j2\pi q(\frac{d_\text{r}}{\lambda})\sin(\cdot)}\right]^\text{T}|_{q=1:N_\text{A}^l}.
\end{equation}

To optimize deployment and antenna resource for distributed APs, we derived CRBs 
for localization and absolute velocity estimation under cooperative sensing. 
These CRBs offer fundamental sensing performance metrics, 
establishing the lower bounds for unbiased estimation.

\subsection{Performance Metrics for Cooperative Sensing Scenario}\label{se2B}
The CRBs for localization and absolute velocity estimation in the CF-mMIMO ISAC cooperative sensing scenario are derived as follows.

\subsubsection{CRB of localization}
For the $u$-th target, the delay, angle, and Doppler information carried in the ISAC echo signal are related to unknown estimation vector $\mathbf{t}_u$. 
Therefore, we define a transform vector 
$\mathbf{t}_u^\text{N}=\left[\mathbf{q}, \mathbf{r}, \mathbf{s}\right]$, where
\begin{equation}\label{eq4}
  \begin{aligned} 
\mathbf{q} & = \left[\tau_1^u,\tau_2^u,\cdots,\tau_L^u\right], \quad \mathbf{s}  = \left[f_{\text{D},1}^u,f_{\text{D},2}^u,\cdots,f_{\text{D},L}^u\right], \\
 \mathbf{r} & = \left[\sin(\theta_1^u),\sin(\theta_2^u),\cdots,\sin(\theta_L^u)\right], 
\end{aligned}  
\end{equation}
and give a \hyperref[lemm1]{Lemma 1} as follows.
\begin{lemma}\label{lemm1}
    Given the FIM and Jacobian matrix about $\mathbf{t}_u^\text{N}$, 
    we can use the chain rule to obtain the localization CRB of $\mathbf{t}_u^\text{N}$
    \begin{equation}\label{eq5}
        \mathrm{CRB}_\text{p} = \left[\mathbf{P}\mathbf{J}\left(\mathbf{t}_u^\text{N}\right)\mathbf{P}^\text{T}\right]^{-1},
    \end{equation}
    where $\mathbf{J}\left(\cdot\right)\in\mathbb{R}^{3L\times 3L}$ is a FIM and $\mathbf{P}\in\mathbb{R}^{2\times 3L}$ is a Jacobian matrix~\cite{godrich2010target}.
\end{lemma}

Using \hyperref[lemm1]{Lemma 1}, the CRB of localization is given by the following theorem.
\begin{theorem}\label{theorem1}
    For the $u$-th target, the CRB of localization in CF-mMIMO ISAC cooperative sensing scenario is 
    \begin{equation}\label{eq6}
\mathrm{CRB}_\mathrm{p}\left(\mathbf{v}_u,\mathbf{t}_u\right) = \left[\mathbf{P}\left[\begin{array}{ccc}
   \mathbf{A}  & \mathbf{B} & \mathbf{C} \\
   \mathbf{B}  & \mathbf{D} & \mathbf{E} \\
   \mathbf{C}  & \mathbf{E} & \mathbf{F}
\end{array}\right]\mathbf{P}^\text{T}\right]^{-1},
    \end{equation}
    where $\mathbf{A}$, $\mathbf{B}$, $\mathbf{C}$, $\mathbf{D}$, $\mathbf{E}$, and $\mathbf{F}$ are all 
    the $L \times L$ diagonal matrices whose diagonal elements are expressed in \eqref{eq7} - \eqref{eq12}, respectively.
    The close form of $\mathbf{P}$ is expressed in (\ref{eq13}) shown at the top of this page, 
    where $\eta_l^{u}=\sqrt{\left(\Theta_l^u\right)^2+\left(\epsilon_l^{u}\right)^2}$, $\epsilon_l^{u}=\left(y_u^{\text{tr}}-y_{l}\right)$, $\Theta_l^u =\left(x_u^{\text{tr}}-x_{l}\right)$, and  $\Xi_l^{u}=f_{\text{D},l}^{u}
\tan\left(\theta_l^{u}-\theta_u\right)$. 
    \begin{equation}\label{eq7}
       {\fontsize{9}{9}
 \left[\mathbf{A}\right]_{l,l} =\frac{4\pi^2\left|\alpha_l^{u}\right|^2\Delta f^2(2N_\text{c}+1)(N_\text{c}+1)N_\text{c}MN_\text{A}^l}{6\sigma_\text{z}^2}, } 
    \end{equation}
    \begin{equation}\label{eq8}
     {\fontsize{9}{9} 
 \left[\mathbf{B}\right]_{l,l} =\frac{-4\pi^2\left|\alpha_l^{u}\right|^2\frac{d_\text{r}}{\lambda}\Delta f(N_\text{A}^l+1)N_\text{A}^l(N_\text{c}+1)N_\text{c}M}{4\sigma_\text{z}^2}, }   
    \end{equation}
    \begin{equation}\label{eq9}
      {\fontsize{9}{9}
 \left[\mathbf{C}\right]_{l,l} =\frac{-4\pi^2\left|\alpha_l^{u}\right|^2T\Delta f(M+1)M(N_\text{c}+1)N_\text{c}N_\text{A}^l}{4\sigma_\text{z}^2}, }  
    \end{equation}
    \begin{equation}\label{eq10}
        {\fontsize{9}{9}
          \left[\mathbf{D}\right]_{{l},{l}}=\frac{4\pi^2\left|\alpha_l^{u}\right|^2\left(\frac{d_\text{r}}{\lambda}\right)^2(2N_\text{A}^l+1)(N_\text{A}^l+1)N_\text{A}^lN_\text{c}M}{6\sigma_\text{z}^2},}
    \end{equation}
    \begin{equation}\label{eq11}
      \left[\mathbf{E}\right]_{{l},{l}} =\frac{4\pi^2\left|\alpha_l^{u}\right|^2\frac{d_\text{r}}{\lambda}T(N_\text{A}^l+1)N_\text{A}^l(M+1)MN_\text{c}}{4\sigma_\text{z}^2},   
    \end{equation}
    \begin{equation}\label{eq12}
     \left[\mathbf{R}\right]_{{l},{l}}=\frac{4\pi^2\left|\alpha_l^{u}\right|^2T^2(2M+1)(M+1)MN_\text{A}^lN_\text{c}}{6\sigma_\text{z}^2}.  
    \end{equation}
    \begin{figure*}[htbp]
    \renewcommand{\arraystretch}{2.5} 
    \begin{equation}\label{eq13}
         \mathbf{P}=\frac{\partial\mathbf{t}_u^\text{N}}{\mathbf{t}_u}=\left[\begin{array}{c@{\hspace{3pt}}c@{\hspace{3pt}}cc@{\hspace{3pt}}c@{\hspace{3pt}}cc@{\hspace{3pt}}c@{\hspace{3pt}}cc@{\hspace{3pt}}c@{\hspace{3pt}}c}
         \dfrac{2\Theta_1^u} {\eta_1^{u}c} & 
         \dfrac{2\Theta_2^u} {\eta_2^{u}c} &
         \cdots &
         \dfrac{2\Theta_L^u} {\eta_L^{u}c} &
         \dfrac{\epsilon_1^{u}\Theta_1^u}{-(\eta_1^{u})^3} &
         \dfrac{\epsilon_2^{u}\Theta_2^u}{-(\eta_2^{u})^3} &
         \cdots&
         \dfrac{\epsilon_L^{u}\Theta_L^u}{-(\eta_L^{u})^3} &
         \dfrac{\Xi_1^{u}\epsilon_1^{u}}{(\eta_1^{u})^2}&
         \dfrac{\Xi_2^{u}\epsilon_2^{u}}{(\eta_2^{u})^2}&
         \cdots&
         \dfrac{\Xi_L^{u}\epsilon_L^{u}}{(\eta_L^{u})^2} 

         \\ \dfrac{2\epsilon_1^{u}}{\eta_1^{u}c}&
         \dfrac{2\epsilon_2^{u}}{\eta_2^{u}c}&
         \cdots&
         \dfrac{2\epsilon_L^{u}}{\eta_L^{u}c}&
         \dfrac{\left(\Theta_1^u\right)^{2}}{\left(\eta_1^{u}\right)^{3}}&
         \dfrac{\left(\Theta_2^u\right)^{2}}{\left(\eta_2^{u}\right)^{3}}&
         \cdots&
         \dfrac{\left(\Theta_L^u\right)^{2}}{\left(\eta_L^{u}\right)^{3}}&
         \dfrac{-\Xi_1^{u}}{\Theta_1^u}&
         \dfrac{-\Xi_2^{u}}{\Theta_2^u}&
         \cdots&
         \dfrac{-\Xi_L^{u}}{\Theta_L^u}\\
        \end{array}\right].   
    \end{equation}
{\noindent} \rule[-10pt]{18cm}{0.1em}
\end{figure*}
\end{theorem}
\begin{proof}
    Please refer to \hyperref[apA]{Appendix A}
\end{proof}
s
\subsubsection{CRB of absolute velocity estimation}
For the $u$-th target, the delay, angle, 
and Doppler information are related to the unknown estimate vector $\mathbf{v}_u $. 
The CRB of absolute velocity estimation differs from that for localization only in the Jacobian matrix. 
Therefore, we derive the following theorem.
\begin{theorem}\label{theorem2}
    For the $u$-th target, the CRB of absolute velocity estimation in CF-mMIMO ISAC cooperative sensing scenario is 
    \begin{equation}\label{eq14}
\mathrm{CRB}_\text{a}\left(\mathbf{v}_u,\mathbf{t}_u\right) = \left[\mathbf{T}\left[\begin{array}{ccc}
   \mathbf{A}  & \mathbf{B} & \mathbf{C} \\
   \mathbf{B}  & \mathbf{D} & \mathbf{E} \\
   \mathbf{C}  & \mathbf{E} & \mathbf{F}
\end{array}\right]\mathbf{T}^\text{T}\right]^{-1}, 
    \end{equation}
    where the close form of $\mathbf{T}$ is expressed in (\ref{eq15}), 
    $\varpi_l^u = \cos\left(\theta_l^u-\theta_u\right)$, and $\chi_l^u = \sin\left(\theta_l^u-\theta_u\right)$.
    \begin{figure*}
      \renewcommand{\arraystretch}{2.5}
      \begin{equation}\label{eq15}
       \mathbf{T} =\frac{\partial\mathbf{t}_u^\text{N}}{\mathbf{v}_u} =\left[\begin{array}{c@{\hspace{1.5pt}}c@{\hspace{1.5pt}}cc@{\hspace{1.5pt}}c@{\hspace{1.5pt}}cc@{\hspace{1.5pt}}c@{\hspace{1.5pt}}cc@{\hspace{1.5pt}}c@{\hspace{1.5pt}}c}
       \dfrac{\epsilon_1^u\Theta_1^u f_{\text{D},1}^u}{\eta_1^u \Xi_1^u c} &
       \dfrac{\epsilon_2^u\Theta_2^u f_{\text{D},2}^u}{\eta_2^u \Xi_2^u c} &
       \cdots &
       \dfrac{\epsilon_L^u\Theta_L^u f_{\text{D},L}^u}{\eta_L^u \Xi_L^u c} &
       \dfrac{\left(\varpi_1^u\right)^2}{v_u\chi_1^u} &
       \dfrac{\left(\varpi_2^u\right)^2}{v_u\chi_2^u} &
        \cdots &
       \dfrac{\left(\varpi_L^u\right)^2}{v_u\chi_L^u} &
       \dfrac{-2f_\text{c}\varpi_1^u}{c} &
       \dfrac{-2f_\text{c}\varpi_2^u}{c} &
       \cdots &
       \dfrac{-2f_\text{c}\varpi_L^u}{c} 
       \\
       \dfrac{\epsilon_1^u\Theta_1^u}{2 \eta_1^u c} &
       \dfrac{\epsilon_2^u\Theta_2^u}{2 \eta_2^u c} &
       \cdots &
       \dfrac{\epsilon_L^u\Theta_L^u}{2 \eta_L^u c} &
       \cos\left(\theta_1^u\right) &
       \cos\left(\theta_2^u\right) &
       \cdots &
       \cos\left(\theta_L^u\right) &
       \dfrac{-2f_\text{c} v_u \chi_1^u}{c} &
       \dfrac{-2f_\text{c} v_u \chi_2^u}{c} &
       \cdots &
       \dfrac{-2f_\text{c} v_u \chi_L^u}{c} \\
        \end{array}\right]. 
      \end{equation}
{\noindent} \rule[-10pt]{18cm}{0.1em}
    \end{figure*}
    
\end{theorem}
\begin{proof}
    Based on the relationship in (\ref{eqap6}), we derive the close form of $\mathbf{T}$ expressed in (\ref{eq15}) shown at the top of this page.
\end{proof}

To improve cooperative sensing performance in the CF-mMIMO ISAC system, 
we propose a two-phase cooperative sensing framework. 
\textit{Phase 1:} Jointly optimizes the placement and antenna resource for distributed APs to lower the sensing CRBs,
providing better system parameters for cooperative sensing, including the placement and the number
of antennas for distributed APs (please refer to Section~\ref{se3}). 
\textit{Phase 2:} an SL-MDTS scheme is presented to achieve precise sensing of multiple targets (see Section~\ref{se4}).

\section{A Joint Optimization Scheme for Sensing Performance}\label{se3}
To address the challenge of joint optimization in the CF-mMIMO ISAC system,
we formulate a joint placement and antenna resource optimization scheme
for distributed APs to lower the sensing CRBs~\cite{willem2012minimax}, 
solved using ADMM~\cite{boyd2011distributed}, analytical methods, 
convex optimization (CVX) toolbox~\cite{grant2014cvx}, and the truncated Newton method~\cite{nocedal1999numerical}.

\subsection{Problem Formulation}\label{se3A}
Using the CRB criterion, the objective function is the weighted sum of CRBs 
for localization and absolute velocity estimation, 
with a weighting factor balancing their contributions.
For convenience, we define an optimization variable $\mathbf{z}=\left[\mathbf{c}_1^\text{T},\mathbf{c}_2^\text{T},\cdots,\mathbf{c}_L^\text{T},N_\text{A}^1,N_\text{A}^2,\cdots,N_\text{A}^L\right]^\text{T}$, and the CRBs of localization and absolute velocity estimation are denoted by $\mathbf{B}_\text{p}\left(\mathbf{v}_u,\mathbf{t}_u,\mathbf{z}\right)$ and $\mathbf{B}_\text{a}\left(\mathbf{v}_u,\mathbf{t}_u,\mathbf{z}\right)$, respectively. 
Given that APs need proper dispersion to ensure sufficient coverage, 
the feasible location of the $l$-th AP is constrained within a region centered at $\hat{\mathbf{c}}_l$ with radius $\varepsilon$, i.e., $\|\mathbf{c}_l-\hat{\mathbf{c}}_l\|_2^2 \le \varepsilon^2$. 
To maximize the sensing performance of target area, 
we formulate a joint optimization problem based on the minimax criterion
\begin{subequations}\label{eq16}
\begin{align}
  \left(\text{P0}\right) \quad & \underset{\mathbf{z}}{\min} \quad \underset{\{\mathbf{t}_u,\mathbf{v}_u\}_{u=1}^U}{\max} \quad \left\{\begin{array}{l}
      \alpha \psi_\text{p} \text{tr}\left(\mathbf{B}_\text{p}\left(\mathbf{v}_u,\mathbf{t}_u,\mathbf{z}\right)\right)  \\
       +(1-\alpha) \psi_\text{a} \text{tr}\left(\mathbf{B}_\text{a}\left(\mathbf{v}_u,\mathbf{t}_u,\mathbf{z}\right)\right)
  \end{array}\right\} \nonumber \\
  &\quad \text{s.t.}\quad\quad\quad \|\mathbf{c}_l-\hat{\mathbf{c}}_l\|_2^2 \le \varepsilon^2, \forall l, \label{eq16a}\\
   & \quad\quad\quad\quad \quad \sum_{l=1}^L N_\text{A}^l \le N_\text{A}^\text{max}, \label{eq16b}\\
   &\quad \quad N_\text{A}^l \in \mathbb{Z}^+, \forall l, 
    \quad N_\text{A}^l \in \left[1, N_\text{A}^\text{max}\right), \forall l.\label{eq16c}
\end{align}
\end{subequations}
where $\alpha\in\left(0,1\right)$ denotes the balancing factor; 
$\psi_\text{p}$ and $\psi_\text{a}$ are the weight factor used to unify the order of magnitude; 
(\ref{eq16a}) is the feasible scope constraint of AP's location; 
(\ref{eq16b}) is the total antenna resource constraint; 
(\ref{eq16c}) describes integer and non-negative constraints for antenna.

To transform problem \hyperref[eq16]{(P0)} into a standard form, 
we introduce an auxiliary variable $\varpi_u$
\begin{equation}\label{eq17}
  \varpi_u = \left\{\begin{array}{l}
      \alpha \psi_\text{p} \text{tr}\left(\mathbf{B}_\text{p}\left(\mathbf{v}_u,\mathbf{t}_u,\mathbf{z}\right)\right)  \\
       +(1-\alpha) \psi_\text{a} \text{tr}\left(\mathbf{B}_\text{a}\left(\mathbf{v}_u,\mathbf{t}_u,\mathbf{z}\right)\right)
  \end{array}\right\},  \forall l 
\end{equation}
Then, we redefine the optimization variables $\mathbf{c}_l = \mathbf{C}_l^\text{t}\mathbf{z}$ and $N_\text{A}^l=\mathbf{c}_l^\text{A}\mathbf{z}$, 
and the auxiliary variables $\mathbf{a}=\mathbf{C}_l^\text{t}\mathbf{z}-\hat{\mathbf{c}}_l$ and $b_l = \mathbf{c}_l^\text{A}\mathbf{z}$, 
where $\mathbf{C}_l^\text{t}$ and $\mathbf{c}_l^\text{A}$ are selection matrix and vector, respectively. 
Therefore, the problem \hyperref[eq16]{P0} is transformed into
\begin{subequations}\label{eq18}
\begin{align}
   &\left(\text{P1}\right) \quad \underset{\mho,\mathbf{z},\{\varpi_u\}_{u=1}^U,\{\mathbf{a}_l,b_l\}_{l=1}^L}{\min} \quad \mho   \nonumber \\
   & \text{s.t.} \quad \quad \mathbf{a}=\mathbf{C}_l^\text{t}\mathbf{z}-\hat{\mathbf{c}}_l, \forall l, \\
   &\quad \quad \quad \  b_l = \mathbf{c}_l^\text{A}\mathbf{z},\forall l, \\
   &\quad \quad \quad \quad \quad (\ref{eq17}) \\
   &\quad \quad \quad \ \|\mathbf{a}_l\|_2^2 \le \varepsilon^2, \forall l, \\
   & \sum_{l=1}^{L}b_l \le N_\text{A}^\text{max}, \quad b_l \in \left[1, N_\text{A}^\text{max}\right), \forall l \label{eq18e}\\
   &\quad \quad \quad \ \varpi_u \le \mho, \forall u,\label{eq18f}
\end{align}
\end{subequations}
where the NP-hard constraint (\ref{eq16c}) is relaxed as a continuous variable. 
Problem \hyperref[eq18]{(P1)} is nonlinear and non-convex, 
resulting in slow convergence and poor stability with conventional methods such as gradient descent and Newton’s method. 
To address this problem, we employ the ADMM framework to decompose the problem into subproblems, 
combining analytical and truncated Newton methods for efficient large-scale optimization.

\subsection{Optimization Problem Solving}\label{se3B}
First, we form an augmented Lagrangian function derived from the optimization problem \hyperref[eq18]{(P1)}
\begin{equation}\label{eq19}
\begin{aligned}
   &L\left(\mho, \mathbf{z}, \{\varpi_u, \gamma_u\}_{u=1}^U, \{\mathbf{a}_l, b_l, \boldsymbol{\lambda}_l, \chi_l\}_{l=1}^L\right)\\
   &= \mho + \mathbf{d}\left(\mathbf{z}, \{\boldsymbol{\lambda}_l,\mathbf{a}_l\}_{l=1}^L\right) + \mathbf{e}\left(\mathbf{z}, \{\chi_l, b_l\}_{l=1}^L\right) \\ &\quad + \mathbf{f}\left(\mathbf{z}, \{\gamma_u,\varpi_u\}_{u=1}^U\right),
\end{aligned}
\end{equation}
where $\mathbf{d}$, $\mathbf{e}$, and $\mathbf{f}$ are expressed in (\ref{eq20}), (\ref{eq21}), and (\ref{eq22}), respectively; $\{\rho_1, \rho_2, \rho_3\} > 0$ are penalty parameters, 
updated using a classical adaptively adjustment scheme~\cite{boyd2011distributed}; $\{\boldsymbol{\lambda}_l, \chi_l\}_{l=1}^L, \{\gamma_u \}_{u=1}^U$ are Lagrangian multipliers.
\begin{equation}\label{eq20}
\begin{aligned}
  &\mathbf{d}\left(\mathbf{z}, \{\boldsymbol{\lambda}_l,\mathbf{a}_l\}_{l=1}^L\right) \\ & = \sum_l^L\left(\left \langle \boldsymbol{\lambda}_l, \left(\mathbf{a}_l-\mathbf{C}_l^\text{t}\mathbf{z}+\hat{\mathbf{c}}_l\right) \right \rangle + \frac{\rho_1}{2}\|\mathbf{a}_l-\mathbf{C}_l^\text{t}\mathbf{z}+\hat{\mathbf{c}}_l\|_2^2\right),    
\end{aligned}
\end{equation}
\begin{equation}\label{eq21}
{\fontsize{9}{9}\begin{aligned}
\mathbf{e}\left(\mathbf{z}, \{\chi_l, b_l\}_{l=1}^L\right) = \sum_l^L\left(\chi_l\left(b_l-\mathbf{c}_l^\text{A}\mathbf{z}\right)+\frac{\rho_2}{2}\left(b_l-\mathbf{c}_l^\text{A}\mathbf{z}\right)^2\right),
\end{aligned} } 
\end{equation}
\begin{figure*}[!ht]
\begin{equation}\label{eq22}
 \mathbf{f}\left(\mathbf{z}, \{\gamma_u,\varpi_u\}_{u=1}^U\right) = \sum_u^U\left\{\begin{array}{l}
   \gamma_u\left(\varpi_u-\alpha \psi_\text{p} \text{tr}\left(\mathbf{B}_\text{p}\left(\mathbf{v}_u,\mathbf{t}_u,\mathbf{z}\right)\right) - (1-\alpha) \psi_\text{a} \text{tr}\left(\mathbf{B}_\text{a}\left(\mathbf{v}_u,\mathbf{t}_u,\mathbf{z}\right)\right)\right)  \\
    + \frac{\rho_3}{2}\left(\varpi_u-\alpha \psi_\text{p} \text{tr}\left(\mathbf{B}_\text{p}\left(\mathbf{v}_u,\mathbf{t}_u,\mathbf{z}\right)\right) - (1-\alpha) \psi_\text{a} \text{tr}\left(\mathbf{B}_\text{a}\left(\mathbf{v}_u,\mathbf{t}_u,\mathbf{z}\right)\right)\right)^2
 \end{array}\right\} .  
\end{equation} 
{\noindent} \rule[-10pt]{18cm}{0.1em}
\end{figure*}
\hspace{0.3em} We decompose (\ref{eq19}) into the following four subproblems for iterative resolution, 
adjusting through Lagrange multipliers, where the $k$-th iteration steps are as follows.

\textbf{\textit{Step 1:}}
(Update $\{\mathbf{a}_l^k\}_{l=1}^L$, 
and fix other variables $\mho^k$, $\mathbf{z}^k$, $\{\varpi_u^k\}_{u=1}^U$, $\{b_l^k\}_{l=1}^L$) 
The $L$ variables are updated sequentially, 
with $\mathbf{a}_1^k$ updated while fixing $\{\mathbf{a}_l^k\}_{l=2}^L$, 
and similarly for the remaining variables. 
Therefore, the update problem for $\mathbf{a}_l^k$ is
\begin{subequations}\label{eq23}
\begin{align}
  \left(\text{P2.1}\right) \quad \mathbf{a}_l^{k+1} &= \arg\underset{\mathbf{a}_l^k}{\min} \left(\begin{array}{l}
    \left \langle \boldsymbol{\lambda}_l^k, \left(\mathbf{a}_l^k-\mathbf{C}_l^\text{t}\mathbf{z}^k  +\hat{\mathbf{c}}_l\right) \right \rangle \\ + \frac{\rho_1}{2}\|\mathbf{a}_l^k-\mathbf{C}_l^\text{t}\mathbf{z}^k+\hat{\mathbf{c}}_l\|_2^2  \end{array}\right) \nonumber \\ & \text{s.t.} \quad\quad\quad \|\mathbf{a}_l^k\|_2^2 \le \varepsilon^2.\label{eq23a}
\end{align}
\end{subequations}
Problem \hyperref[eq23]{(P2.1)} is convex with inequality constraints and 
can be solved using gradient computation and projection operations. 
The updated $\mathbf{a}_l^{k+1}$ is expressed as
\begin{equation}\label{eq24}
  \mathbf{a}_l^{k+1} = \text{Proj}_{\|\mathbf{a}_l^k\|_2^2 \le \varepsilon^2} \left(\mathbf{C}_l^\text{t}\mathbf{z}^k-\hat{\mathbf{c}}_l-\frac{\boldsymbol{\lambda}_l^k}{\rho_1}\right), 
\end{equation}
where $\text{Proj}_{\|\mathbf{a}_l^k\|_2^2 \le \varepsilon^2}$ denotes that 
the updated optimal variable is projected onto constraint (\ref{eq23a}). 
The detailed derivation of the analytic solution is given in \hyperref[apex2]{Appendix B}.

\textbf{\textit{Step 2:}}
(Update $\{b_l^k\}_{l=1}^L$, and fix other variables $\mho^k$, $\mathbf{z}^k$, $\{\varpi_u^k\}_{u=1}^U$, $\{\mathbf{a}_l^{k+1}\}_{l=1}^L$)
Constraint (\ref{eq18e}) couples the $U$ variables, preventing independent optimization. 
Therefore, the update problem is
\begin{subequations}\label{eq25}
    \begin{align}
      \left(\text{P2.2}\right)  &\underset{\{b_l^k\}_{l=1}^L}{\min} \sum_l^L\left(\chi_l^k\left(b_l^k-\mathbf{c}_l^\text{A}\mathbf{z}^k\right)+\frac{\rho_2}{2}\left(b_l^k-\mathbf{c}_l^\text{A}\mathbf{z}^k\right)^2\right)  \nonumber \\
     & \text{s.t.}  \quad\sum_{l=1}^{L}b_l^k \le N_\text{A}^\text{max}, \quad b_l^k \in \left[1, N_\text{A}^\text{max}\right), \forall l,
    \end{align}
\end{subequations}
which is solved via the CVX toolbox~\cite{grant2014cvx}.

\textit{Step 3:}\label{se3B3}
(Update $\mho^k, \{\varpi_u^k\}_{u=1}^U$, and fix other variables $\mathbf{z}^k$, $\{b_l^{k+1}$, $\mathbf{a}_l^{k+1}\}_{l=1}^L$)
The update problem is expressed as
\begin{subequations}\label{eq26}
    \begin{align}
        \left(\text{P2.3}\right) \quad & \underset{\mho^k,\{\varpi_u^k\}_{u=1}^U}{\min}  \quad \mho^k+ \mathbf{f}\left(\mathbf{z}^k, \{\gamma_u^k,\varpi_u^k\}_{u=1}^U\right) \nonumber \\
        & \quad \quad \text{s.t.} \quad\quad\quad \varpi_u^k \le \mho^k, \forall u.
    \end{align}
\end{subequations}
For problem \hyperref[eq26]{(P2.3)}, we intend to update $\mho^k$ and $\{\gamma_u^k,\varpi_u^k\}_{u=1}^U$ successively, and the specific update steps are as follows:
\begin{figure*}
    \begin{subequations}\label{eq27}
        \begin{align}
            \left(\text{P2.3.1}\right) \quad \varpi_u^{k+1} = \arg\underset{\varpi_u^k}{\min} & \left\{\begin{array}{l}
   \gamma_u^k\left(\varpi_u^k-\alpha \psi_\text{p} \text{tr}\left(\mathbf{B}_\text{p}\left(\mathbf{v}_u,\mathbf{t}_u,\mathbf{z}^k\right)\right) - (1-\alpha) \psi_\text{a} \text{tr}\left(\mathbf{B}_\text{a}\left(\mathbf{v}_u,\mathbf{t}_u,\mathbf{z}^k\right)\right)\right)  \\
    + \frac{\rho_3}{2}\left(\varpi_u^k-\alpha \psi_\text{p}\text{tr}\left(\mathbf{B}_\text{p}\left(\mathbf{v}_u,\mathbf{t}_u,\mathbf{z}^k\right)\right) - (1-\alpha) \psi_\text{a}\text{tr}\left(\mathbf{B}_\text{a}\left(\mathbf{v}_u,\mathbf{t}_u,\mathbf{z}^k\right)\right)\right)^2
 \end{array}\right\} \nonumber \\
    & \text{s.t.} \quad\quad\quad \varpi_u^k \le \mho^k, \forall u.
        \end{align}
    \end{subequations}
{\noindent} \rule[-10pt]{18cm}{0.1em}
\end{figure*}
\begin{itemize}
    \item Fix $\mho^k$, and update $\{\gamma_u^k,\varpi_u^k\}_{u=1}^U$. 
    The update problem of $\varpi_u^k$ is represented in \hyperref[eq27]{(P2.3.1)} shown at the top of this page, 
    which has a similar solution to \hyperref[eq23]{(P2.1)}, and the close form of $\{\varpi_u^{k+1}\}_{u=1}^U$ is 
\end{itemize}
    \begin{equation}\label{eq28}
  {\fontsize{9}{8}  \begin{aligned}
       \varpi_u^{k+1} =  \min\left(\mho^k,\left\{\begin{array}{l}
          \alpha \psi_\text{p} \text{tr}\left(\mathbf{B}_\text{p}\left(\mathbf{v}_u,\mathbf{t}_u,\mathbf{z}^k\right)\right) + \\
           (1-\alpha) \psi_\text{a} \text{tr}\left(\mathbf{B}_\text{a}\left(\mathbf{v}_u,\mathbf{t}_u,\mathbf{z}^k\right)\right) - \frac{\gamma_u^k}{\rho_3} 
       \end{array}\right\}\right) .
    \end{aligned}  }  
    \end{equation}
\begin{itemize}
    \item Fix $\{\varpi_u^{k+1}\}_{u=1}^U$, and update $\mho^k$. Based on the definition of $\mho^k$, the updated variable is $\mho^{k+1} = \underset{u}{\max}\left(\varpi_u^{k+1}\right)$.
\end{itemize}


\textbf{\textit{Step 4:}}
(Update $\mathbf{z}^k$, and fix other variables $\mho^{k+1}
$, $\{\mathbf{a}_l^{k+1},b_l^{k+1}\}_{l=1}^L$, $\{\varpi_u^{k+1}\}_{u=1}^U$) The update problem of $\mathbf{z}^k$ can be expressed in (\ref{eq29}) shown at the top of this page.
\begin{figure*}
    \begin{align}\label{eq29}
     \left(\text{P2.4}\right)  \quad \mathbf{z}^{k+1} & = \arg\underset{\mathbf{z}^k}{\min} \quad f\left(\mathbf{z}^k\right) \nonumber \\
        & = \arg\underset{\mathbf{z}^k}{\min} \quad \left\{ \begin{array}{c}
          \sum_l^L \frac{\rho_1}{2}\|\mathbf{a}_l^{k+1}-\mathbf{C}_l^\text{t}\mathbf{z}^k+\hat{\mathbf{c}}_l+\frac{\boldsymbol{\lambda}_l^k}{\rho_1}\|_2^2 + \sum_l^L\frac{\rho_2}{2}\left(b_l^{k+1}-\mathbf{c}_l^\text{A}\mathbf{z}^k+\frac{\chi_l^k}{\rho_2}\right)^2     \\
          + \sum_u^U\frac{\rho_3}{2}\left(\varpi_u^{k+1}-\alpha\psi_\text{p}\text{tr}\left(\mathbf{B}_\text{p}\left(\mathbf{v}_u,\mathbf{t}_u,\mathbf{z}^k\right)\right)-(1-\alpha)\psi_\text{a}\text{tr}\left(\mathbf{B}_\text{a}\left(\mathbf{v}_u,\mathbf{t}_u,\mathbf{z}^k\right)\right)+\frac{\gamma_u^k}{\rho_3}\right)^2            
        \end{array} \right\}.
    \end{align}
{\noindent} \rule[-10pt]{18cm}{0.1em}
\end{figure*}
Problem \hyperref[eq29]{(P2.4)} is an unconstrained, nonlinear, nonconvex problem with a high-dimensional variable $\mathbf{z}$, solved using the truncated Newton method~\cite{nocedal1999numerical}, ideal for large-scale optimization.

\textbf{\textit{Step 5:}}
Update variables $\{\boldsymbol{\lambda}_l^k, \chi_l^k\}_{l=1}^L, \{\gamma_u^k\}_{u=1}^U$
\begin{equation}\label{eq30}
 {\fontsize{8.5}{9} \begin{array}{c}
      \boldsymbol{\lambda}_l^{k+1} = \boldsymbol{\lambda}_l^{k} + \rho_1 \left(\mathbf{a}_l^{k+1}-\mathbf{C}_l^\text{t}\mathbf{z}^{k+1}+\hat{\mathbf{c}}_l\right), \forall l \\
      \chi_l^{k+1} = \chi_l^{k} + \rho_2\left(b_l^{k+1}-\mathbf{c}_l^\text{A}\mathbf{z}^{k+1}\right), \forall l \\
      \gamma_u^{k+1} = \gamma_u^{k} + \rho_3 \left(\begin{array}{l}
          \varpi_u^{k+1}-\alpha\psi_\text{p} \text{tr}\left(\mathbf{B}_\text{p}\left(\mathbf{v}_u,\mathbf{t}_u,\mathbf{z}^{k+1}\right)\right)  \\
           (1-\alpha) \psi_\text{a} \text{tr}\left(\mathbf{B}_\text{a}\left(\mathbf{v}_u,\mathbf{t}_u,\mathbf{z}^{k+1}\right)\right) 
       \end{array}\right).\forall u
  \end{array} }
\end{equation}

The four subproblems and the Lagrange multiplier updates are iterated until convergence. Primal residuals are defined as $\zeta_\text{p}^k= \underset{l}{\max}\frac{\|\mathbf{a}_l^k-\mathbf{C}_l^\text{t}\mathbf{z}^k+\hat{\mathbf{c}_l}\|_2}{\|\mathbf{a}_l^k\|_2}$, $\phi_\text{p}^k = \underset{l}{\max}\frac{|b_l^k-\mathbf{c}_l^k\mathbf{z}^k|}{|b_l^k|}$, and $\upsilon_\text{p}^k = \underset{u}{\max}\frac{|\varpi_u^k-\alpha\psi_\text{p}\text{tr}\left(\mathbf{B}_\text{p}\right)-(1-\alpha)\psi_\text{a}\text{tr}\left(\mathbf{B}_\text{a}\right)|}{|\varpi_u^k|}$. The dual residuals are $\zeta_\text{d}^k =\underset{l}{\max}\frac{\|\mathbf{a}_l^k-\mathbf{a}_l^{k-1}\|_2}{\|\mathbf{a}_l^k\|_2}$, $\phi_\text{d}^k = \underset{l}{\max}\frac{|b_l^k-b_l^{k-1}|}{|b_l^k|}$, and $\upsilon_\text{d}^k = \underset{l}{\max}\frac{|\varpi_u^k-\varpi_u^{k-1}|}{|\varpi_u^k|}$. The joint deployment and antenna resource optimization scheme for distributed APs is shown in \hyperref[alg2]{Algorithm 2}.
The optimized location and antenna number of the $l$-th AP are $\tilde{\mathbf{c}}_l = \mathbf{C}_l^\text{t}\mathbf{z}^*$ and $\tilde{N}_\text{A}^l = \mathbf{c}_l^\text{A}\mathbf{z}^*$, respectively.

\begin{table}[!ht]
\centering
\label{alg2}
\resizebox{0.9\linewidth}{!}{
\setlength{\arrayrulewidth}{1.5pt}
\begin{tabular}{rllll}
\hline
\multicolumn{5}{l}{\textbf{Algorithm 2:} Proposed Joint Optimization Scheme}   \\ \hline
\multirow{-3}{*}{\textbf{Input:} }               & \multicolumn{4}{l}{\begin{tabular}[c]{@{}l@{}}$\{\mathbf{t}_u, \mathbf{v}_u\}_{u=1}^U$, $\{\mathbf{C}_l^\text{t}, \mathbf{c}_l^\text{A}\}_{l=1}^L$, $N_\text{A}^\text{max}$, and $\{\hat{\mathbf{c}}_l\}_{l=1}^L$;\\ Maximum iterative time $K_\text{max}$ and radius $\varepsilon$;\\ Feasibility tolerances $\hat{\varepsilon}_\zeta$, $\hat{\varepsilon}_\phi$, and $\hat{\varepsilon}_\upsilon$. \end{tabular}} \\
\textbf{Output:}               & \multicolumn{4}{l}{The optimized variable $\mathbf{z}^{*}$.} \\ 
1:       & \multicolumn{4}{l}{Initialize $\mathbf{z}^0, \{\boldsymbol{\lambda}_l^0,\chi_l^0\}_{l=1}^L, \{\gamma_u^0\}_{u=1}^U, \rho_1^0,\rho_2^0,\rho_3^0$;} \\
2:        & \multicolumn{4}{l}{\textbf{for} $k = 0,1,\cdots,K_\text{max}$ \textbf{do}}  \\
3:     & \multicolumn{4}{l}{$\hspace{0.8em}$ Obtain $\{\mathbf{a}_l^{k+1}\}_{l=1}^L$ by solving problem \hyperref[eq23]{(P2.1)};}  \\
4:     & \multicolumn{4}{l}{$\hspace{0.8em}$ Obtain $\{b_l^{k+1}\}_{l=1}^L$ by solving problem \hyperref[eq25]{(P2.2)};}  \\
5:  & \multicolumn{4}{l}{$\hspace{0.8em}$ Obtain $\{\varpi_u^{k+1}\}_{l=1}^L$ and $\mho^{k+1}$ by \hyperref[se3B3]{\textit{Step 3}};}  \\
6:                              & \multicolumn{4}{l}{$\hspace{0.8em}$ Obtain $\mathbf{z}^{k+1}$ by using truncated Newton method;}  \\
7:                           & \multicolumn{4}{l}{$\hspace{0.8em}$ Update $\{\boldsymbol{\lambda}_l^{k+1},\chi_l^{k+1}\}_{l=1}^L, \{\gamma_u^{k+1}\}_{u=1}^U$ by (\ref{eq30});}   \\
8:                              & \multicolumn{4}{l}{$\hspace{0.8em}$ \textbf{if} the convergence conditions are satisfied \textbf{do}}  \\
9:                            & \multicolumn{4}{l}{$\hspace{1.6em}$ return $\mathbf{z}^{*} = \mathbf{z}^k$;}  \\
10:                              & \multicolumn{4}{l}{$\hspace{0.8em}$ \textbf{else} Update $\rho_1,\rho_2,\rho_3$ by the scheme in~\cite{boyd2011distributed} ;}  \\
11:                              & \multicolumn{4}{l}{\textbf{end for} }  \\
\hline
\end{tabular}}
\end{table}

Then, we introduce a cooperative sensing signal processing scheme to achieve the effective fusion of multiple sensing information and precise sensing of multiple targets, 
as detailed in Section~\ref{se4}. 

\section{A Cooperative Sensing Signal Processing Scheme}\label{se4}
As shown in Fig. \ref{fig2}, an SL-MDTS scheme is proposed to address the challenge of 
sensing information fusion and achieve high-accuracy sensing of multiple targets, 
which comprises two stages: signal preprocessing and symbol-level fusion-based sensing.
\begin{figure*}
    \centering \includegraphics[width=0.95\textwidth]{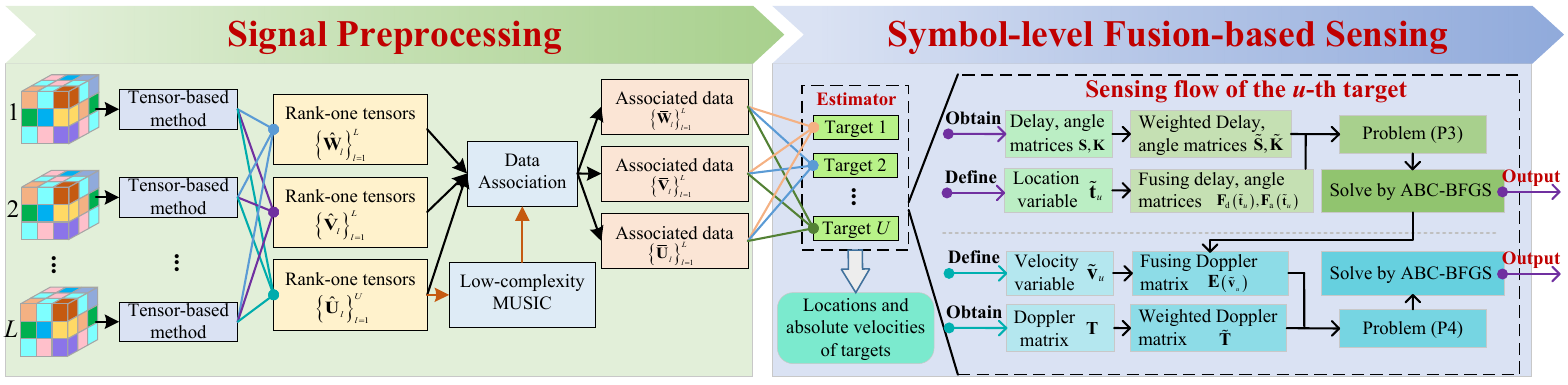}
    \caption{The diagram of the proposed SL-MDTS scheme}
    \label{fig2}
\end{figure*}

\subsection{Signal Preprocessing Stage}\label{sec3-a}
From (\ref{eq1}), the sensing information of $U$ dynamic targets is coupled at each AP, 
requiring decoupling for multi-target sensing. 
Additionally, each AP lacks knowledge of which target corresponds to the received sensing information, 
necessitating data association to establish a one-to-one mapping between information and targets.
Therefore, the signal preprocessing stage is further subdivided into the following two steps.
\begin{enumerate}
    \item \textbf{\textit{Decoupling for multiple targets:}} The decoupling of the sensing information for $U$ dynamic targets at each AP is achieved by exploiting structures of high-dimensional echo signals and performing tensor decomposition~\cite{kolda2009tensor}.
    \item \textbf{\textit{Data association:}} The one-to-one relationship is established using an exhaustive search method for each AP, transforming the multi-target sensing problem into $U$ individual single-target sensing problems.
\end{enumerate}

\subsubsection{Decoupling for multiple targets}
According to (\ref{eq1}), for the $l$-th AP, 
the third-order tensor form of the ISAC echo signal of the $l$-th AP canceling of 
known pilot data on $N_\text{c}$ subcarriers during $M$ OFDM symbol periods 
is~\cite{kolda2009tensor}
\begin{equation}\label{eq31}
     \mathcal{Y}_l\approx \left[\!\left[ \mathbf{U}_l,\mathbf{V}_l,\mathbf{W}_l \right]\!\right] + \mathcal{Z}\equiv \sum_{u=1}^{U}\mathbf{a}_{\text{a},l}^u\circ  \mathbf{a}_{\text{d},l}^u\circ   \mathbf{a}_{\text{f},l}^u+ \mathcal{Z}_l,
\end{equation}
where $\mathcal{Y}_l \in \mathbb{C}^{\tilde{N}_\text{A}^l\times N_\text{c}\times M}$, 
and $\mathcal{Z}_l$ is an AWGN tensor; $\mathbf{U}_l=\left[\mathbf{a}_{\text{a},l}^1,\mathbf{a}_{\text{a},l}^2,\cdots, \mathbf{a}_{\text{a},l}^U\right]$, $\mathbf{V}_l=\left[\mathbf{a}_{\text{d},l}^1,\mathbf{a}_{\text{d},l}^2,\cdots, \mathbf{a}_{\text{d},l}^U\right]$, and $\mathbf{W}_l=\left[\mathbf{a}_{\text{f},l}
^1,\mathbf{a}_{\text{f},l}
^2,\cdots, \mathbf{a}_{\text{f},l}
^U\right]$; $\mathbf{a}_{\text{a},l}^u$, $\mathbf{a}_{\text{d},l}^u$, and $\mathbf{a}_{\text{f},l}^u$ are the rank-one tensors of the $u$-th target, expressed in (\ref{eq32}), (\ref{eq33}), and (\ref{eq34}), respectively.
\begin{equation}\label{eq32}
\mathbf{a}_{\text{a},l}^u=\sqrt{P_\mathrm{t}}\mathbf{a}_\text{r}\left(\theta_l^u\right)\chi_\text{t}^u, 
\end{equation}
\begin{equation}\label{eq33}
\mathbf{a}_{\text{d},l}^u=\left[e^{-j2\pi\Delta f\tau_l^u}, e^{-j4\pi\Delta f\tau_l^u}, \cdots, e^{-j2\pi N_\text{c}\Delta f\tau_l^u}\right]^\text{T}, 
\end{equation}
\begin{equation}\label{eq34}
 \mathbf{a}_{\text{f},l}^u= \left[e^{j2\pi f_{\text{D},l}^uT}, e^{j4\pi f_{\text{D},l}^uT}, \cdots, e^{j2\pi M f_{\text{D},l}^uT}\right]^\text{T}.  
\end{equation}

The \cite{liu2024multipath} prove that the factor matrices of $\mathcal{Y}_l$ satisfy Kruskal's condition, guaranteeing unique CANDECOMP/PARAFAC (CP) decomposition. 
Based on (\ref{eq31}), the CANDECOMP/PARAFAC (CP) decomposition problem is expressed in~\cite{liu2024multipath}
\begin{equation}\label{eq35}
    \min_{\mathbf{\hat{U}}_l, \mathbf{\hat{V}}_l, \mathbf{\hat{W}}_l}\left\|\mathcal{Y}_l-\left[\!\left[ \mathbf{\hat{U}}_l, \mathbf{\hat{V}}_l, \Hat{\mathbf{W}}_l \right]\!\right]\right\|_{\text{F}},
\end{equation}
which is solved by the regularized alternating least-squares method~\cite{kolda2009tensor}.

Repeating (\ref{eq31})-(\ref{eq35}) for each AP, 
the decoupling of multi-dynamic target sensing information is achieved. 
For the $l$-th AP, decomposition yields $\hat{\mathbf{U}}_l$, $\hat{\mathbf{V}}_l$, and $\hat{\mathbf{W}}_l$.
It should be noted that each AP lacks knowledge of which target corresponds to the received sensing information.
To avoid ambiguity, the index of the decomposed rank-one tensor is renumbered, and the decomposition products for the $l$-th AP are expressed as
\begin{equation}\label{eq36}
\begin{aligned}
\hat{\mathbf{U}}_l=\left[\hat{\mathbf{a}}_{\text{a},l}^\xi\right]|_{\xi=1:U},
\hat{\mathbf{V}}_l=\left[\hat{\mathbf{a}}_{\text{d},l}^\xi\right]|_{\xi=1:U},
\hat{\mathbf{W}}_l=\left[\hat{\mathbf{a}}_{\text{f},l}^\xi\right]|_{\xi=1:U},
\end{aligned}
\end{equation}
which $\xi$ denotes the relative index, serving as a convenient reference for identification, while $u$ represents the absolute index, corresponding to the globally unified target identity.

\subsubsection{Data association}\label{sce4-A2}
Eq.~(\ref{eq36}) indicates that it is unclear which target corresponds to the three rank-one tensors $\hat{\mathbf{a}}_{\text{a},l}^\xi$, $\hat{\mathbf{a}}_{\text{d},l}^\xi$, and $\hat{\mathbf{a}}_{\text{f},l}^\xi$.
Therefore, a data association operation is applied to establish the relationship between the three rank-one tensors and the targets. The detailed steps are as follows.

\textit{\textbf{Step 1:}}
Process $\{\hat{\mathbf{a}}_{\text{d},l}^\xi\}_{\xi=1}^U$ to estimate the distances between the $U$ targets and the $l$-th AP. 
To minimize the impact of estimation errors on data association while maintaining low complexity, 
we propose a variant of the low-complexity MUSIC method~\cite{wei2024integrated}, 
as described below.

To estimate the distance $R_l^\xi$ between the $\xi$-th target and the $l$-th AP, we first perform an inverse discrete Fourier transform (IDFT) on $\left[(\hat{\mathbf{a}}_{\text{d},l}^\xi)^\text{T}, 0_{1 \times (N_\text{ifft} - N_\text{c})  }\right]$, where $N_\text{ifft}$ denotes the sampling number of IDFT. 
The rough estimate is $\tilde{R}_l^\xi = \frac{\hat{n}c}{2\Delta f N_\text{ifft}}$, 
where $\hat{n}$ is the peak index in the IDFT output. 
Thus, the MUSIC method’s search interval is denoted by
\begin{equation}\label{eq37}
    R_\text{search} = \left[\frac{(\hat{n}-1)c}{2\Delta f N_\text{ifft}},\quad \frac{(\hat{n}+1) c}{2\Delta f N_\text{ifft}}\right].
\end{equation}
Using $R_\text{search}$, we apply the MUSIC method to obtain a high-accuracy estimation. 

Computing the covariance matrix $\mathbf{R}_\text{cm} = \hat{\mathbf{a}}_{\text{d},l}^\xi(\hat{\mathbf{a}}_{\text{d},l}^\xi)^\text{H}$ and decomposing it via eigenvalue decomposition (EVD), we obtain
\begin{equation}\label{eq38}
\text{eig}\left(\mathbf{R}_\text{cm}\right) = \left[\mathbf{U}_\text{em}, \Sigma_\text{ev}\right],  
\end{equation}
where $\Sigma_\text{ev}$ is the eigenvalue diagonal matrix, and $\mathbf{U}_\text{em}$ is the corresponding eigenvector matrix. The noise subspace $\mathbf{U}_\text{n}=\mathbf{U}_\text{em}\left(:, 2:N_\text{c}\right)$ is used to construct 
\begin{equation}\label{eq39}
    \mathbf{s}(\hat{r}) = \frac{1}{\mathbf{a}(\hat{r})^\text{H}\mathbf{U}_\text{n}\mathbf{U}_\text{n}^\text{H}\mathbf{a}(\hat{r})},
\end{equation}
where $\hat{r} \in R_\text{search}$, $\mathbf{a}(\hat{r})=\left[e^{-j2\pi n \Delta f \frac{2\hat{r}}{c}}\right]^\text{T}|_{n = 1:N_\text{c}}$, and 
the peak index of $\mathbf{s}$ is the final distance estimate $\hat{R}_l^\xi$.

Repeating (\ref{eq37}) - (\ref{eq39}) for $\{\hat{\mathbf{a}}_{\text{d},l}^\xi\}_{\xi=1}^U$, the distance estimates for the $l$-th AP are denoted as $\{\hat{R}_l^\xi\}_{\xi=1}^U$.

\textbf{\textit{Step 2:}} $\{\hat{R}_l^\xi\}_{\xi=1}^U$ is reshaped as a vector $\mathbf{r}_l=\left[R_l^\xi\right]|_{\xi=1:U} = \mathbf{d}_l\mathbf{Q}_l$, where $\mathbf{d}_l=\left[\hat{r}_l^u\right]|_{u=1:U}$ is the range vector associating the distance information with the $U$ targets, and $\mathbf{Q}_l\in\mathbb{C}^{U\times U}$ is a permutation matrix obtained via the traditional exhaustive search method~\cite{Liu_sensing}. 

\textbf{\textit{Step 3:}} Use the obtained $\mathbf{Q}_l$, the decomposition products after data association are given by \eqref{eq40}.
\begin{figure*}
\begin{equation}\label{eq40}
\tilde{\mathbf{U}}_l=\hat{\mathbf{U}}_l\mathbf{Q}_l = \left[\hat{\mathbf{a}}_{\text{a},l}^u\right]|_{u=1:U}, 
\quad\quad
\tilde{\mathbf{V}}_l=\hat{\mathbf{V}}_l\mathbf{Q}_l = \left[\hat{\mathbf{a}}_{\text{d},l}^u\right]|_{u=1:U},
\quad\quad
\tilde{\mathbf{W}}_l=\hat{\mathbf{W}}_l\mathbf{Q}_l = \left[\hat{\mathbf{a}}_{\text{f},l}^u\right]|_{u=1:U}.
\end{equation}  
{\noindent} \rule[-10pt]{18cm}{0.1em}
\end{figure*}

Repeating \textit{\textbf{Step 1-3}} to each AP's decomposition products yields $\{\tilde{\mathbf{U}}_l,\tilde{\mathbf{V}}_l,\tilde{\mathbf{W}}_l\}_{l=1}^L$ for the second stage, transforming the multi-target sensing problem into $U$ single-target sensing problems. Without loss of generality, the subsequent stage focuses on the $u$-th target's sensing process.

\subsection{Symbol-Level Fusion-based Sensing Stage}\label{sce4-B}
We propose an SFO method for localization and velocity estimation, formulating the sensing problem as an optimization to achieve fusion gains and continuous estimation.

\subsubsection{Localization}
$\{\hat{\mathbf{a}}_{\text{d},l}^u\}_{l=1}^L$ and $\{\hat{\mathbf{a}}_{\text{a},l}^u\}_{l=1}^L$ are used for localization. 
For $\{\hat{\mathbf{a}}_{\text{d},l}^u\}_{l=1}^L$, the $l$-th element can be expressed as 
$\hat{\mathbf{a}}_{\text{d},l}^u = \mathbf{a}_{\text{d},l}^u + \omega_l^u$ 
with $\omega_l^u$ being a noise term. 
We construct a delay matrix for symbol-level fusion, which is expressed as 
\begin{equation}\label{eq41}
    \mathbf{S} = \left[\hat{\mathbf{a}}_{\text{d},1}^u,\hat{\mathbf{a}}_{\text{d},2}^u,\cdots,\hat{\mathbf{a}}_{\text{d},L}^u\right]^\text{T}\in\mathbb{C}^{L\times N_\text{c}}.
\end{equation}
The noise variance of each $\hat{\mathbf{a}}_{\text{d},l}^u$, used to determine the weight factor, is estimated as $\hat{\sigma}_l^2$ via EVD method~\cite{liu2024noise}. Based on maximal ratio combining principle, the weighted delay matrix is $\tilde{\mathbf{S}} = \mathbf{S} \odot \left(\mathbf{w}_\text{f} \otimes \mathbf{i}_{1\times N_\text{c}}\right) $, where $\mathbf{i}_{1\times N_\text{c}}$ is a vector of all ones, $\mathbf{w}_\text{f}$ is a weight vector, expressed as
\begin{equation}\label{eq42}
   \mathbf{w}_\text{f} = \left[\frac{\frac{1}{\hat{\sigma}_1^2}}{\sum_l^L\frac{1}{\hat{\sigma}_l^2}},\frac{\frac{1}{\hat{\sigma}_2^2}}{\sum_l^L\frac{1}{\hat{\sigma}_l^2}},\cdots,\frac{\frac{1}{\hat{\sigma}_L^2}}{\sum_l^L\frac{1}{\hat{\sigma}_l^2}}\right]^\text{T} \in \mathbb{R}^{L\times 1}.
\end{equation}

For $\{\hat{\mathbf{a}}_{\text{a},l}^u\}_{l=1}^L$, we pad the dimensions of the $L$ rank-one tensors with zeros to ensure consistency under nonuniform antenna allocation. The $l$-th padded element can be expressed as 
$\hat{\mathbf{a}}_{\text{s},l}^u =\left[\left[\mathbf{a}_{\text{a},l}^u +\omega_l^u\right]^\text{T},0_{\Psi-N_\text{A}^l}\right]^\text{T}$ with     
$\Psi = \underset{l}{\max}\left(N_\text{A}^l\right)$. 
We construct a angle matrix for symbol-level fusion, which is
\begin{equation}\label{eq43}
   \mathbf{K} = \left[\hat{\mathbf{a}}_{\text{s},1}^u,\hat{\mathbf{a}}_{\text{s},2}^u,\cdots,\hat{\mathbf{a}}_{\text{s},L}^u\right]^\text{T}\in\mathbb{C}^{L\times \Psi}.
\end{equation}
Since tensor decomposition has the same effect on three-rank tensors, $\mathbf{w}_\text{f}$ can be shared for $\mathbf{K}$, and the weighted angle matrix is $\tilde{\mathbf{K}}= \mathbf{K}\odot \left(\mathbf{w}_\text{f} \otimes \mathbf{i}_{1\times \Psi}\right)$.

To transform the localization problem into an optimization problem, we first define an optimization variable $\tilde{\mathbf{t}}_u = \left[\tilde{x}_u^\text{tr},\tilde{y}_u^\text{tr}\right]^\text{T}$. 
Then, the fusing delay vector for the $l$-th AP is 
\begin{equation}\label{eq44}
  \mathbf{f}_l^\text{d} = \left[e^{j2\pi n\Delta f\frac{2\|\tilde{\mathbf{t}}_u-\tilde{\mathbf{c}}_l\|_2}{c}}\right]^\text{T}|_{n=1:N_\mathrm{c}},  
\end{equation}
and the fusing delay matrix is 
\begin{equation}\label{eq45}
  \mathbf{F}_\text{d}(\tilde{\mathbf{t}}_u) = \left[\mathbf{f}_1^\text{d},\mathbf{f}_2^\text{d},\cdots,\mathbf{f}_L^\text{d}\right]\in\mathbb{C}^{N_\text{c}\times L}.
\end{equation}
The fusing angle vector for the $l$-th AP is 
\begin{equation}\label{eq46}
   \mathbf{f}_l^\text{a} = \left[e^{-j2\pi b\frac{d_\text{r}}{\lambda}\sin(\bar{\theta}_l^u)}\right]^\text{T}|_{b=1:\Psi},
\end{equation}
and the fusing angle matrix is 
\begin{equation}\label{eq47}
  \mathbf{F}_\text{a}(\tilde{\mathbf{t}}_u) = \left[\mathbf{f}_1^\text{a},\mathbf{f}_2^\text{a},\cdots,\mathbf{f}_L^\text{a}\right]\in\mathbb{C}^{\Psi\times L},  
\end{equation}
where $\bar{\theta}_l^u = \arctan(\frac{\tilde{\mathbf{t}}_u(2)-\tilde{\mathbf{c}}_l(2)}{\tilde{\mathbf{t}}_u(1)-\tilde{\mathbf{c}}_l(1)})$.

According to (\ref{eq41}) - (\ref{eq47}), the optimization problem for the optimization variable $\tilde{\mathbf{t}}_u$ can be formulated as
\begin{subequations}
\begin{align}\label{eq48}
    \text{(P3)}  \quad &\underset{\tilde{\mathbf{t}}_u}{\min} \quad \frac{1}{\text{tr}\left(|\tilde{\mathbf{S}}\mathbf{F}_\text{d}(\tilde{\mathbf{t}}_u)|+|\tilde{\mathbf{K}}\mathbf{F}_\text{a}(\tilde{\mathbf{t}}_u)|\right)} \nonumber \\
    & \text{s.t.} \quad  \text{tr}\left(|\tilde{\mathbf{S}}\mathbf{F}_\text{d}(\tilde{\mathbf{t}}_u)|+|\tilde{\mathbf{K}}\mathbf{F}_\text{a}(\tilde{\mathbf{t}}_u)|\right) > 0.
\end{align}
\end{subequations}
The surface in Fig. \ref{fig3} illustrates that problem \hyperref[eq48]{(P3)} is non-convex, which can be solved via particle swarm optimization (PSO)~\cite{wang2018particle}, 
differential evolution (DE)~\cite{wang2018particle}, simulated annealing (SA)~\cite{wang2018particle}, genetic algorithm (GA)~\cite{wang2018particle}, and artificial bee colony (ABC)~\cite{karaboga2008performance}. 
Among these, ABC is efficient for medium-scale continuous multi-modal problems such as \hyperref[eq48]{(P3)}~\cite{karaboga2008performance}, 
as demonstrated in Section \ref{se5}. 
For \hyperref[eq48]{(P3)}, we suggest the parameter selection of ABC to obtain best performance (detailed in Section \ref{se5}). 
However, due to ABC's limited local optimization accuracy, 
we further employ a hybrid algorithm, referred to as ABC-broyden-fletcher-goldfarb-shanno (BFGS), where ABC obtains a global approximate solution, 
and BFGS refines it to enhance accuracy. 
The simulation results confirm the superiority of the ABC-BFGS over the ABC.
The approximate solution to the problem \hyperref[eq48]{(P3)} is the estimated location of the $u$-th target  $\tilde{\mathbf{t}}_u^* =\left[\hat{x}_u^\text{tr},\hat{y}_u^\text{tr}\right]^\text{T}$. 

\subsubsection{Absolute velocity estimation}
$\{\hat{\mathbf{a}}_{\text{f},l}^u\}_{l=1}^L$ are used to estimate the absolute velocity of the target. 
The $l$-th element of $\{\hat{\mathbf{a}}_{\text{f},l}^u\}_{l=1}^L$ can be expressed as 
$\hat{\mathbf{a}}_{\text{f},l}^u =\mathbf{a}_{\text{f},l}^u + \omega_l^u$.   
We construct a Doppler matrix for symbol-level fusion, which is
\begin{equation}\label{eq49}
    \mathbf{T} =\left[\hat{\mathbf{a}}_{\text{f},1}^u ,\hat{\mathbf{a}}_{\text{f},2}^u ,\cdots,\hat{\mathbf{a}}_{\text{f},L}^u \right]^\text{T}\in\mathbb{C}^{L\times M}, 
\end{equation}
and the weighted Doppler matrix is $\tilde{\mathbf{T}} = \mathbf{T} \odot \left(\mathbf{w}_\text{f} \otimes \mathbf{i}_{1\times M}\right)$.

Then, we define an optimization variable $\tilde{\mathbf{v}}_u = \left[\tilde{v}_u,\tilde{\theta}_u\right]^\text{T}$. 
The construct of fusing Doppler vector needs the AoAs between $L$ APs and target $u$. 
Due to the difference in accuracy of angle estimation caused by antenna allocation, 
we use the target locations obtained earlier to calculate AoAs. 
The fusing Doppler vector for the $l$-th AP is 
\begin{equation}\label{eq50}
\begin{aligned}
   \mathbf{e}_l^\text{f} = \left[e^{j2\pi m \frac{2f_\text{c}\tilde{\mathbf{v}}(1)\cos(\hat{\theta}_l^u-\tilde{\mathbf{v}}(2))}{c}T}\right]|_{m=1:M},    
\end{aligned}
\end{equation}
where $\hat{\theta}_l^u =\arctan\left(\frac{\hat{y}_u^\text{tr}-\tilde{\mathbf{c}}_l(2)}{\hat{x}_u^\text{tr}-\tilde{\mathbf{c}}_l(1)}\right)$.
The fusing Doppler matrix is 
\begin{equation}\label{eq51}
  \mathbf{E}(\tilde{\mathbf{v}}_u) = \left[\left(\mathbf{e}_1^\text{f}\right)^\text{T},\left(\mathbf{e}_2^\text{f}\right)^\text{T},\cdots,\left(\mathbf{e}_L^\text{f}\right)^\text{T}\right]\in\mathbb{C}^{M\times L}.  
\end{equation}

According to (\ref{eq49}) - (\ref{eq51}), 
the optimization problem for the optimization variable $\tilde{\mathbf{v}}_u$ is 
\begin{subequations}\label{eq52}
\begin{align}
   \text{(P4)}  \quad &\underset{\tilde{\mathbf{v}}_u}{\min} \quad \frac{1}{\text{tr}\left(|\tilde{\mathbf{T}}\mathbf{E}(\tilde{\mathbf{v}}_u)|\right)} \nonumber \\
    & \text{s.t.} \quad  \text{tr}\left(|\tilde{\mathbf{T}}\mathbf{E}(\tilde{\mathbf{v}}_u)|\right) > 0.  
\end{align}
\end{subequations}
Based on the surface in Fig. \ref{fig4}, 
problem \hyperref[eq52]{(P4)} is a nonlinear, non-convex problem that can be solved by
the ABC-BFGS, and the estimated absolute velocity of the $u$-th target 
is $\tilde{\mathbf{v}}_u^* = \left[\hat{v}_u,\hat{\theta}_u\right]^\text{T}$.

\begin{figure}
    \centering   \includegraphics[width=.30\textwidth]{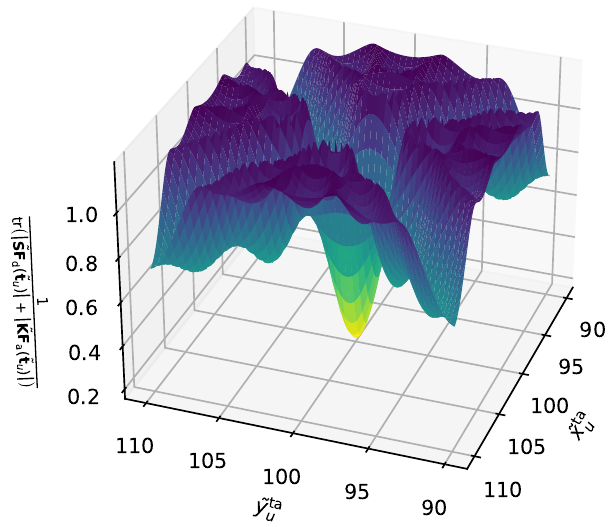}
    \caption{Surface of the objective function in problem \hyperref[eq48]{(P3)}}
    \label{fig3}
\end{figure}
\begin{figure}
    \centering   \includegraphics[width=.30\textwidth]{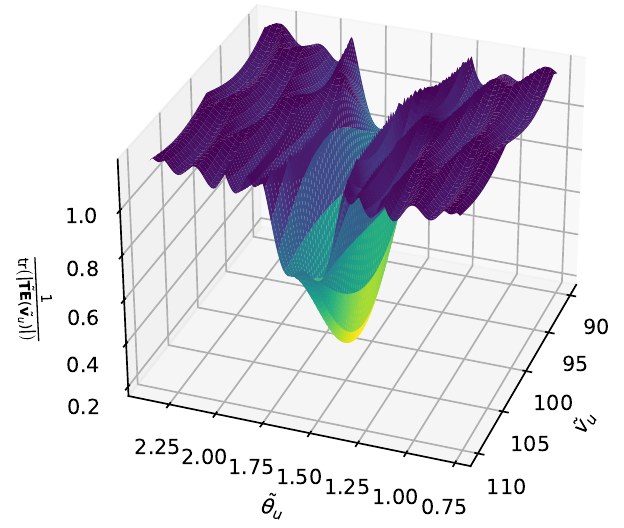}
    \caption{Surface of the objective function in problem \hyperref[eq52]{(P4)}}
    \label{fig4}
\end{figure}

\hyperref[algo3]{Algorithm 3} demonstrates the proposed SL-MDTS scheme.
\begin{table}[!ht]
\centering
\label{algo3}
\resizebox{0.9\linewidth}{!}{
\setlength{\arrayrulewidth}{1.5pt}
\begin{tabular}{rllll}
\hline
\multicolumn{5}{l}{\textbf{Algorithm 3:} Proposed SL-MDTS scheme}   \\ \hline
\multirow{-2}{*}{\textbf{Input:} }               & \multicolumn{4}{l}{\begin{tabular}[c]{@{}l@{}}Received three-order tensors $\{\mathcal{Y}\}_{l=1}^L$ in (\ref{eq31});\\ The locations of APs $\{\tilde{\mathbf{c}}_l\}_{l=1}^L$; \end{tabular}} \\
\multirow{-2}{*}{\textbf{Output:} }               & \multicolumn{4}{l}{\begin{tabular}[c]{@{}l@{}}The estimated target's locations $\{\tilde{\mathbf{t}}_u^*\}_{u=1}^U$;\\
The estimated target's absolute velocities $\{\tilde{\mathbf{v}}_u^*\}_{u=1}^U$.\end{tabular}} \\ 
\multicolumn{5}{l}{\textbf{Signal Preprocessing Stage:}} \\
1:  & \multicolumn{4}{l}{\textbf{for} $l$ in $L$ \textbf{do}}  \\
\multirow{-2}{*}{2:}       & \multicolumn{4}{l}{\begin{tabular}[c]{@{}l@{}}\hspace{0.8em} Obtain  matrices $\hat{\mathbf{U}}_l$, $\hat{\mathbf{V}}_l$, and $\hat{\mathbf{W}}_l$ in (\ref{eq36}) by \\ \hspace{0.8em} CP decomposition, and the method proposed in \cite{liu2024multipath};\end{tabular}} \\
\multirow{-2}{*}{3:}        & \multicolumn{4}{l}{\begin{tabular}[c]{@{}l@{}}\hspace{0.8em} Obtain $\tilde{\mathbf{U}}_l$, $\tilde{\mathbf{V}}_l$, and $\tilde{\mathbf{W}}_l$ in (\ref{eq40}) by data\\ \hspace{0.8em} association operation; \end{tabular}}  \\
4:  & \multicolumn{4}{l}{\textbf{end for}}  \\
\multicolumn{5}{l}{\textbf{Symbol-Level Fusion-based Sensing Stage:}} \\
5:  & \multicolumn{4}{l}{\textbf{for} $u$ in $U$ \textbf{do}}  \\
6:     & \multicolumn{4}{l}{\hspace{0.8em} Obtain the delay matrix $\mathbf{S}$ in (\ref{eq41});}  \\
7:     & \multicolumn{4}{l}{\hspace{0.8em} Obtain weighted delay matrix $\tilde{\mathbf{S}}$ by (\ref{eq41}) and (\ref{eq42});}  \\
8:     & \multicolumn{4}{l}{\hspace{0.8em} Obtain the angle matrix $\mathbf{K}$ in (\ref{eq43});}  \\
9:     & \multicolumn{4}{l}{\hspace{0.8em} Obtain weighted angle matrix $\tilde{\mathbf{K}}$ by (\ref{eq43}) and $\mathbf{w}_\text{f}$;}  \\
10:  & \multicolumn{4}{l}{\hspace{0.8em} Obtain fusing delay matrix $\mathbf{F}_\text{d}\left(\tilde{\mathbf{t}}_u\right)$ in (\ref{eq45});}  \\
11:  & \multicolumn{4}{l}{\hspace{0.8em} Formulate an optimization problem \hyperref[eq48]{(P3)};}  \\
12:  & \multicolumn{4}{l}{\hspace{0.8em} Solve \hyperref[eq48]{(P3)} by ABC-BFGS algorithm to obtain $\tilde{\mathbf{t}}_u^*$;}   \\
13:     & \multicolumn{4}{l}{\hspace{0.8em} Obtain the Doppler feature matrix $\mathbf{T}$ in (\ref{eq49});}  \\
14:     & \multicolumn{4}{l}{\hspace{0.8em} Obtain weighted Doppler matrix $\tilde{\mathbf{T}}$ by (\ref{eq42}) and (\ref{eq49});}  \\
15:  & \multicolumn{4}{l}{\hspace{0.8em} Obtain fusing Doppler matrix $\mathbf{E}\left(\tilde{\mathbf{v}}_u\right)$ in (\ref{eq51});}  \\
16:  & \multicolumn{4}{l}{\hspace{0.8em} Formulate an optimization problem \hyperref[eq52]{(P4)};}  \\
17:  & \multicolumn{4}{l}{\hspace{0.8em} Solve \hyperref[eq52]{(P4)} by ABC-BFGS algorithm to obtain $\tilde{\mathbf{v}}_u^*$;}   \\
18:  & \multicolumn{4}{l}{\textbf{end} \textbf{for}}
\\
\hline
\end{tabular}}
\end{table}

\section{Simulation}\label{se5}
This section presents simulation results to evaluate the performance of the proposed joint optimization scheme and validate the feasibility and advantages of the SL-MDTS scheme. 
In addition, the effectiveness of the two-phase cooperative sensing framework is confirmed through numerical simulations, and the interplay between the two phases is further analyzed.

Unless otherwise specified, the global simulation parameters are listed as follows. 
The carrier frequency is set to $f_\text{c}=3.5$ GHz and the sub-carrier spacing is $\Delta f = 30$ kHz. 
The number of subcarriers and OFDM symbols is $N_\text{c} = 128$ and $M=128$, respectively~\cite{liu2024carrier,liu2024target,3gpp2018nr}.

\subsection{Joint Optimization Scheme}
Without loss of generality, we set $L = 4$, with a total antenna resource of
$N_\text{A}^\text{max} = 32$, an AP feasible radius of $\varepsilon = 10$ m, 
and initial AP positions at (0, 0) m, (350, 0) m, (0, 350) m, and (350, 350) m, respectively. 
Within the target area, spanning from (125, 125) to (225, 225) m, 
the target points are randomly distributed, with each target’s absolute velocity 
randomly assigned from (90, $0^\circ$) to (120, $\pi$).
For the proposed joint optimization scheme in \hyperref[alg2]{Algorithm 2}, $\{\rho_1^0, \rho_2^0, \rho_3^0\}$ are set to $\{10^{-2}, 10^{-2}, 6\times 10^{1}\}$, $\{\hat{\varepsilon}_\zeta,\hat{\varepsilon}_\phi,\hat{\varepsilon}_\upsilon\}$ are set to $\{10^{-8},10^{-8},10^{-8}\}$, $\{\boldsymbol{\lambda}_l^k, \chi_l^k\}_{l=1}^L, \{ \gamma_u^k\}_{u=1}^U$ are random with seed being 2024. 

For Fig.~\ref{fig5}, ``Initial'' denotes
the results with the initial AP positions and antenna equipartition, 
``Random'' represents the results obtained by randomly assigning the AP positions and antenna distribution under constraints, 
and ``Proposed'' reflects the results after optimization. 
The results demonstrate that the proposed method reduces both individual CRB metrics and the overall objective function for the target area, 
validating its effectiveness. 
The results of the ``Random'' show values both below and above those of the ``Initial'', 
suggesting the presence of optimization potential in the system and indirectly reflecting the authenticity of the optimization process.
\begin{figure}
	\centering
	\subfigure[CRB of location estimation] {\label{fig5.a}\includegraphics[width=0.24\textwidth]{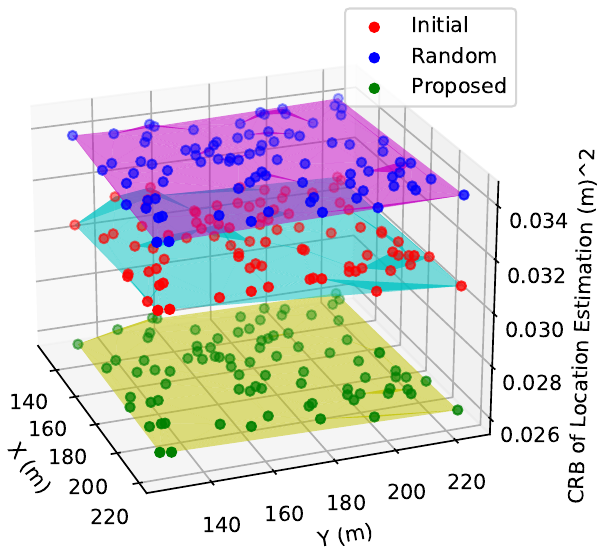}}
    \subfigure[CRB of absolute velocity estimation] {\label{fig5.b}\includegraphics[width=0.24\textwidth]{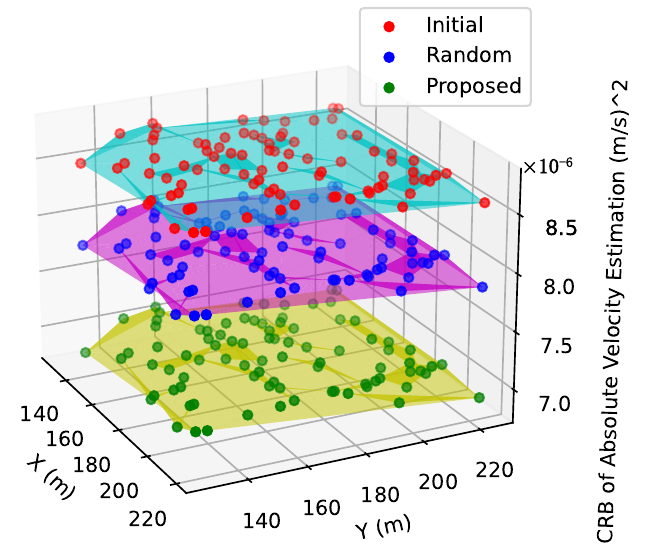}}
    \subfigure[objective function] {\label{fig5.c}\includegraphics[width=0.30\textwidth]{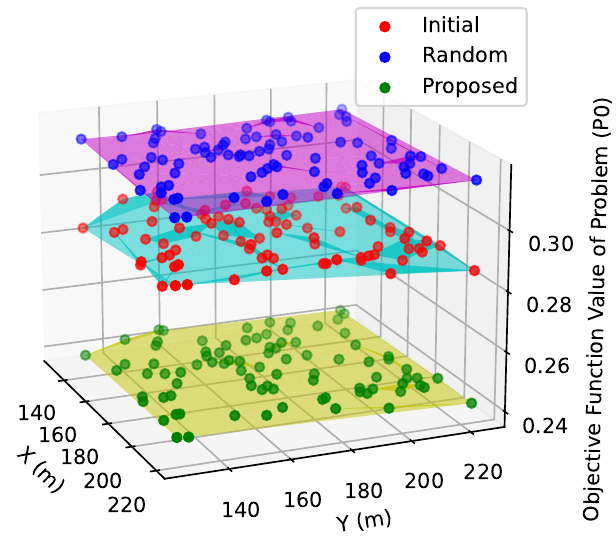}}
	\caption{Initial \textit{v.s.} Random \textit{v.s.} Proposed optimization scheme}
	\label{fig5}
\end{figure}

\subsection{Proposed SL-MDTS scheme}
The feasibility and superiority of the proposed SL-MDTS scheme are validated through comprehensive simulations. 
First, different global optimization algorithms are compared to demonstrate the superiority of the ABC algorithm adopted in this paper. 
The parameter configurations for applying the ABC algorithm are then analyzed through simulations to ensure the efficiency of the scheme. 
Finally, the feasibility and superiority of the proposed SL-MDTS scheme are validated by estimating the location and absolute velocity profiles and benchmarking against state-of-the-art methods.

\subsubsection{Comparison of global optimization algorithms}
For the problems \hyperref[eq48]{(P3)} and \hyperref[eq52]{(P4)}, swarm intelligence algorithms, including DE, SA, GA, PSO, and ABC, are commonly used~\cite{wang2018particle,karaboga2008performance}. 
Figs.~\ref{fig6}(a) and \ref{fig6}(b) demonstrate the root mean square errors (RMSEs) of 
the problem's solutions within various algorithms. 
The setup of the control parameters of DE, SA, GA, PSO and ABC is as follows~\cite{wang2018particle,karaboga2008performance,krink2004noisy}.
\textit{\textbf{DE}}: popsize =50, crossover probability = 0.5, and scaling factor = 0.8; \textit{\textbf{SA}}: initial\_temp = 5230, visit = 2.62, accept = -5, and maxiter = 100; \textit{\textbf{GA}}: popsize = 20, crossover probability = 0.9, mutation probability = 0.3, and maxiter = 50; \textit{\textbf{PSO}}: popsize = 80, inertia weight = 0.5, cognitive component = 2, social component = 2, and maxiter = 200; \textbf{\textit{ABC}}: popsize = 20, epoch = 10, and maxiter = 50.

As shown in Figs.~\ref{fig6}(a) and \ref{fig6}(b), 
ABC and GA exhibit superior performance, followed by PSO, 
while SA and DE demonstrate the poorest results. 
This is attributed to the enhanced robustness of the diversity and search mechanisms 
in ABC and GA against noise, 
whereas the acceptance criteria and mutation operations in SA and DE 
are vulnerable to noise perturbations.

\begin{figure}
	\centering
	\subfigure[RMSE for solving problem (P3)] {\label{fig6.a}\includegraphics[width=0.35\textwidth]{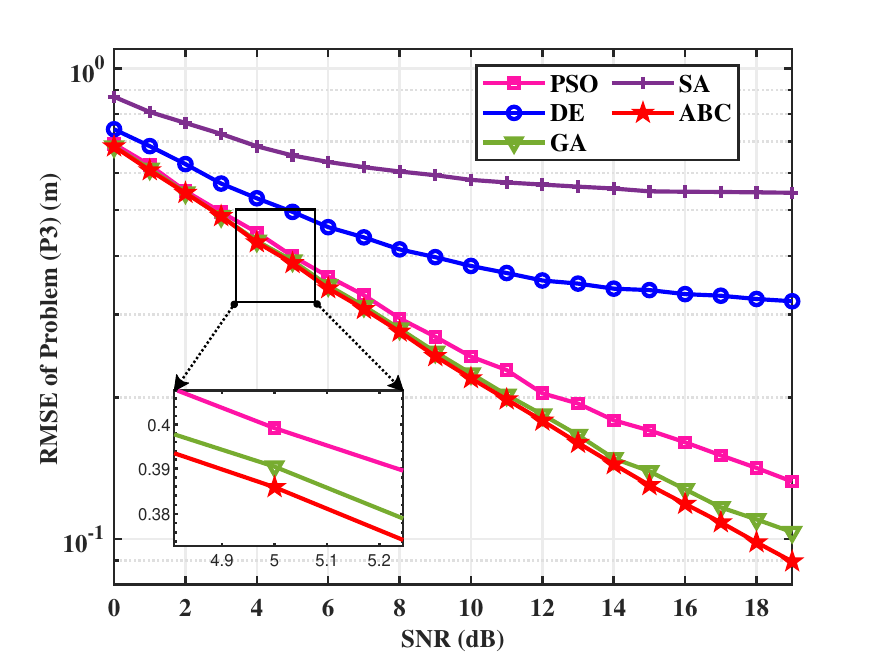}}
	\subfigure[RMSE for solving problem (P4)] {\label{fig6.b}\includegraphics[width=0.35\textwidth]{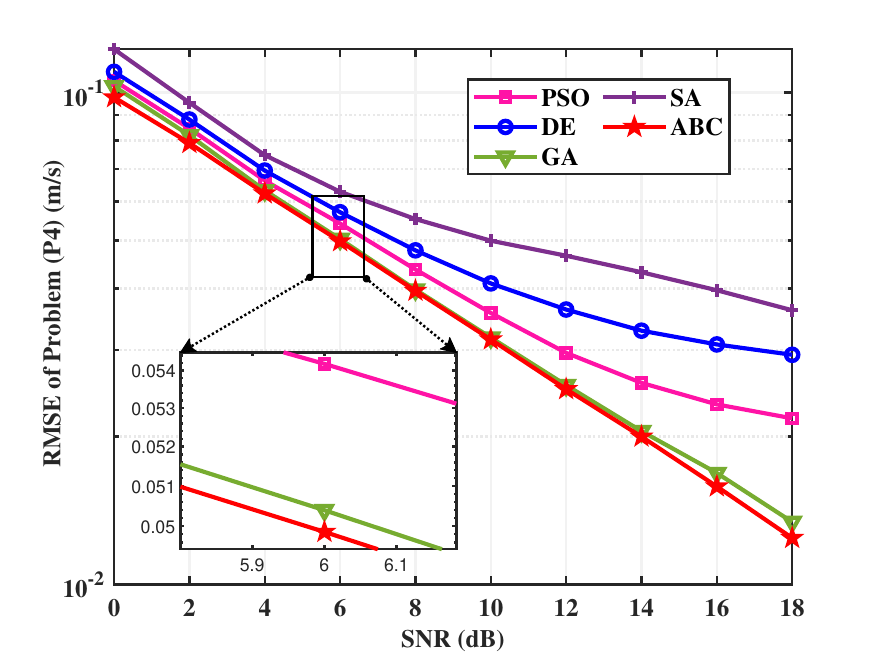}}
	\caption{Comparison of different algorithms for solving the multi-peak continuous optimization problems}
	\label{fig6}
\end{figure}

\subsubsection{Parameter selection of ABC algorithm}
Figs.~\ref{fig7}(a) and \ref{fig7}(b) analyze the relationship between epoch and RMSE for different combinations of $P$ (popsize) and SNR (in dB).

The simulation results show that increasing the popsize (P) under the same SNR (e.g., in cases 1, 2, and 3) accelerates convergence and reduces the number of epochs. 
Conversely, for the same P (e.g., in cases 3, 6, and 9), increasing the SNR reduces the error and accelerates convergence for a given error level. 
Therefore, we can select appropriate epochs and $P$ for different SNRs to ensure both fast convergence and high-precision solutions. 
This insight also motivates the design of adaptive parameters to improve algorithm performance, which will be our future work.
\begin{figure}
	\centering
	\subfigure[RMSE for solving problem (P3)] {\label{fig7.a}\includegraphics[width=0.35\textwidth]{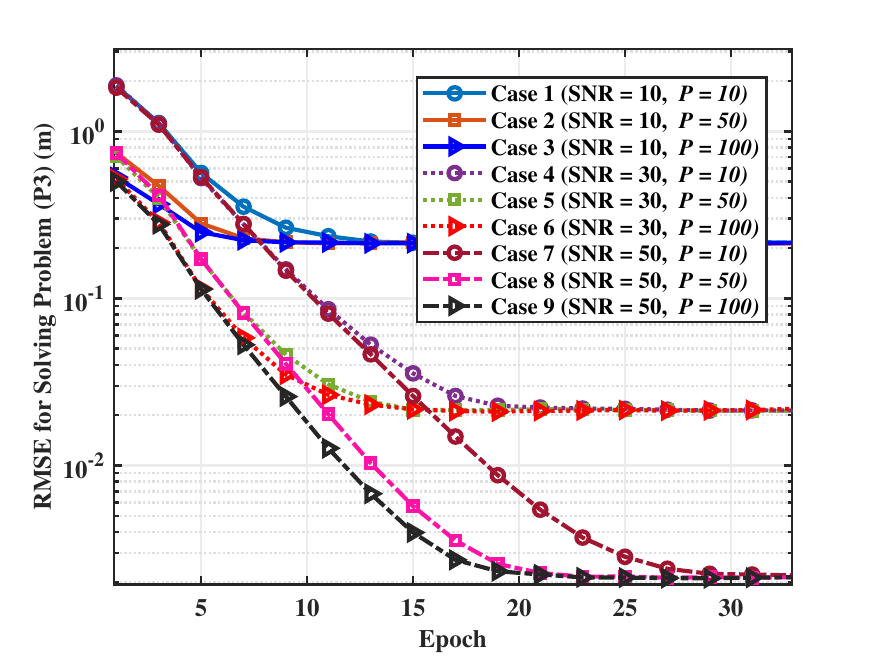}}
	\subfigure[RMSE for solving problem (P4)] {\label{fig7.b}\includegraphics[width=0.35\textwidth]{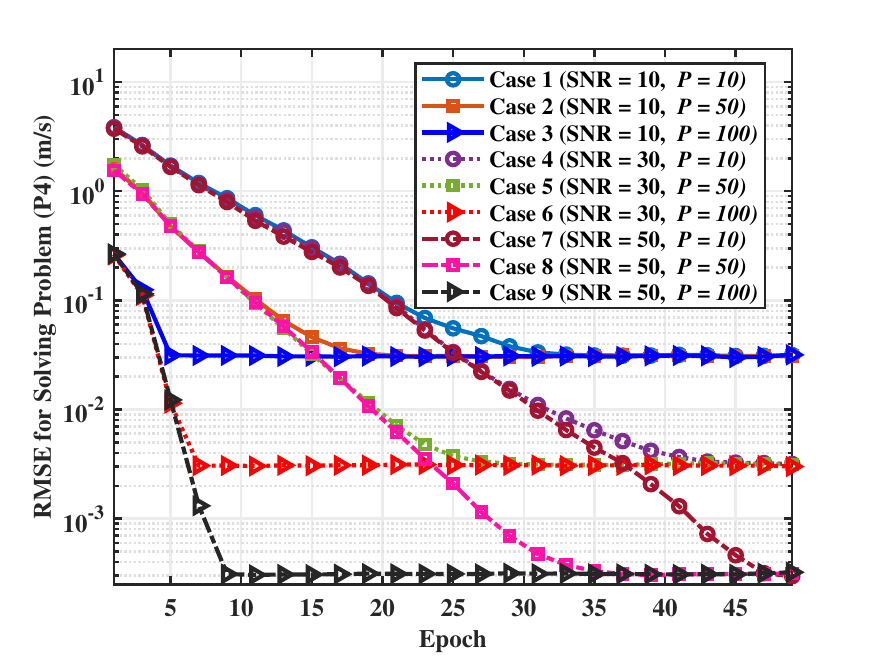}}
	\caption{RMSE \textit{v.s.} epochs for SNR of 10, 30, 50 dB and popsize of 10, 50, 100}
	\label{fig7}
\end{figure}

\subsubsection{SL-MDTS scheme}
We first investigate the feasibility of the proposed SL-MDTS scheme by simulating the estimation profiles of location and absolute velocity, respectively.
The results are shown in Fig.~\ref{fig8}, where the target detection threshold is set to 0.7 and the SNR of the simulation environment is 0 dB.
Observing Figs.~\ref{fig8}(a) and \ref{fig8}(b), the locations and absolute velocities of the three potential targets are estimated with high precision and no ghost targets are present, which demonstrates the feasibility of the proposed SL-MDTS scheme.

\begin{figure}
	\centering
	\subfigure[Location estimation] {\label{fig8.a}\includegraphics[width=0.35\textwidth]{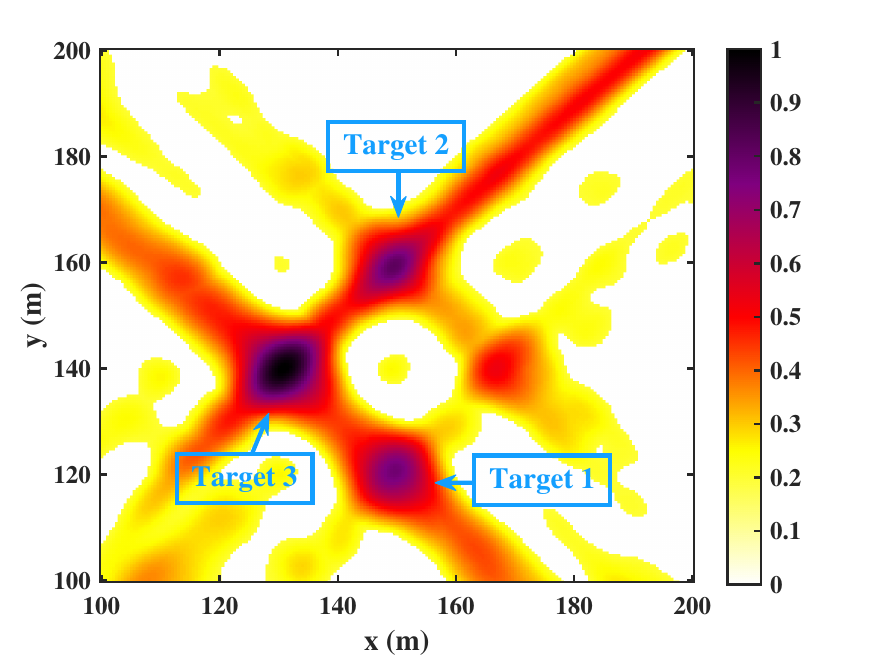}}
	\subfigure[Absolute velocity estimation] {\label{fig8.b}\includegraphics[width=0.35\textwidth]{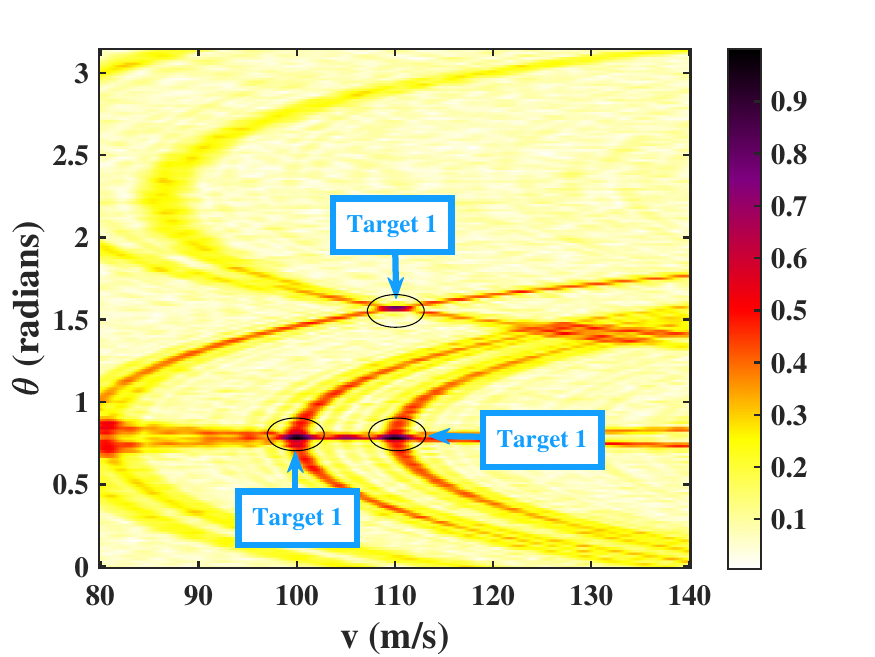}}
	\caption{Estimation profiles of location and absolute velocity with SNR = 0 dB}
	\label{fig8}
\end{figure}

Since absolute velocity estimation in CF-mMIMO ISAC cooperative sensing has not yet been explored in existing work, 
and most benchmark methods focus on data fusion, 
the evaluation emphasizes the proposed SFO method in the second stage, 
while using a uniform first stage process for all benchmarks to ensure a fair comparison.

Performance comparison for multi-dynamic target estimations using the proposed SFO method is demonstrated in Fig.~\ref{fig9}, evaluated in terms of the average RMSE (ARMSE)~\cite{liu2024carrier}. 
The theoretical performance of location and absolute velocity estimations is defined as $\sqrt{\text{CRB}_\text{p}}$ and $\sqrt{\text{CRB}_\text{a}}$.
Our benchmark methods include the MLE method~\cite{Shi_sensing} and the lattice method~\cite{wei_symbol}, where the MLE method jointly processes the estimation results from multiple APs to perform MLE, while the lattice method fuses the sensing information from multiple APs at the symbol level and employs a grid structure for estimation. 
As illustrated in Figs.~\ref{fig9}(a) and \ref{fig9}(b), compared with the two benchmark methods, the proposed SPO method improves the location estimation performance by 50.1\% and 44\%, and the absolute velocity estimation performance by 72.9\% and 41.4\%, respectively. 
This is due to the fact that SPO outperforms MLE by achieving symbol-level SNR gains and avoids the accuracy limitations of lattice methods through continuous spatial search. 
Meanwhile, the simulation results verify the superiority of the applied ABC-BFGS algorithm over ABC.

\begin{figure}
	\centering
	\subfigure[Location estimation] {\label{fig9.a}\includegraphics[width=0.33\textwidth]{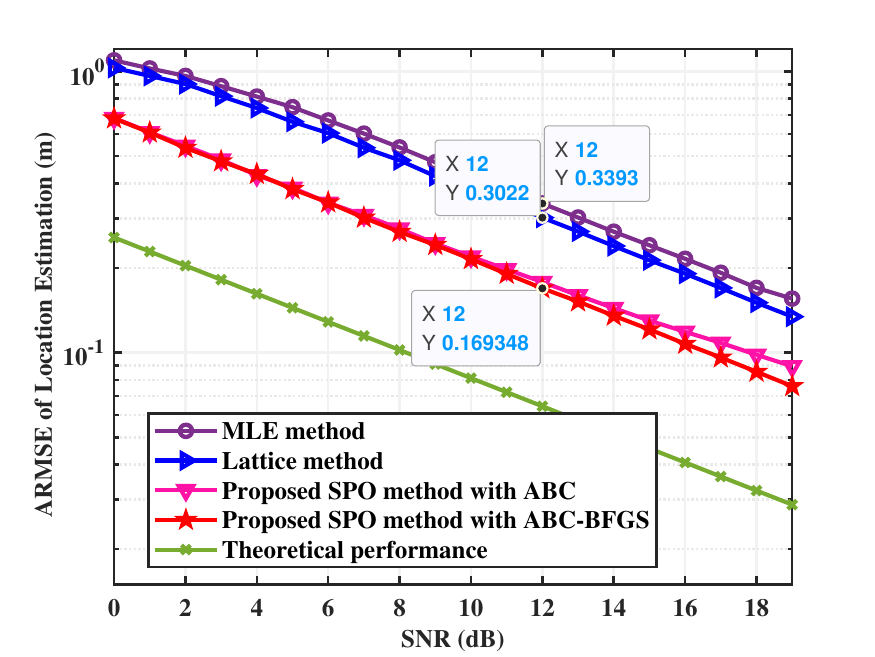}}
	\subfigure[Absolute velocity estimation] {\label{fig9.b}\includegraphics[width=0.33\textwidth]{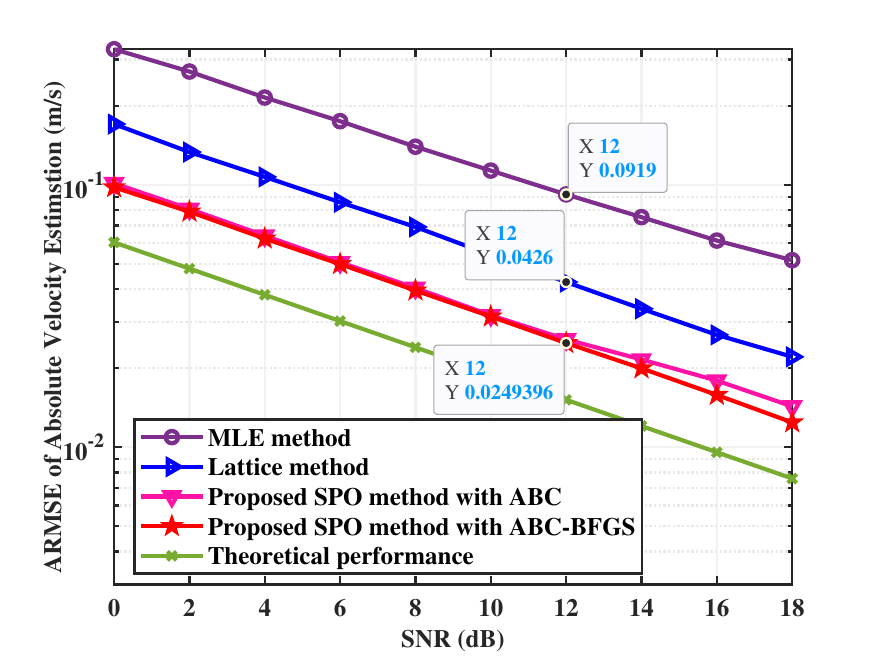}}
	\caption{ARMSEs of location and absolute velocity estimation, with the MLE~\cite{Shi_sensing} and lattice schemes~\cite{wei_symbol} as benchmarks}
	\label{fig9}
\end{figure}

\subsection{Two-phase cooperative sensing framework}
We simulate the ARMSE of multi-dynamic target sensing to explore the contribution of collaboration between the two phases to target sensing, i.e., 
the improvement of theoretical sensing performance in the first stage yields a sensing gain in the second stage.
Since minor improvements in theoretical performance under low SNR conditions are easily overwhelmed by noise-induced disturbances on the ARMSE, 
we conducted simulations under higher SNR conditions to present a more accurate ARMSE, providing a more intuitive view to reveal the contribution of 
collaboration between the two phases.

As shown in Fig.~\ref{fig10}, \textit{Initial} and \textit{Optimal} refer to the sensing ARMSEs based on the initial and optimized system parameters, respectively. 
Compared to \textit{Initial}, the ARMSEs of \textit{Optimal} are lower, demonstrating the contribution of collaboration between the two phases. 
This contribution results from the optimization of the CRB, which further reduces the achievable lower bound of the ARMSE.

\begin{figure}
	\centering
	\subfigure[Location estimation] {\label{fig10.a}\includegraphics[width=0.33\textwidth]{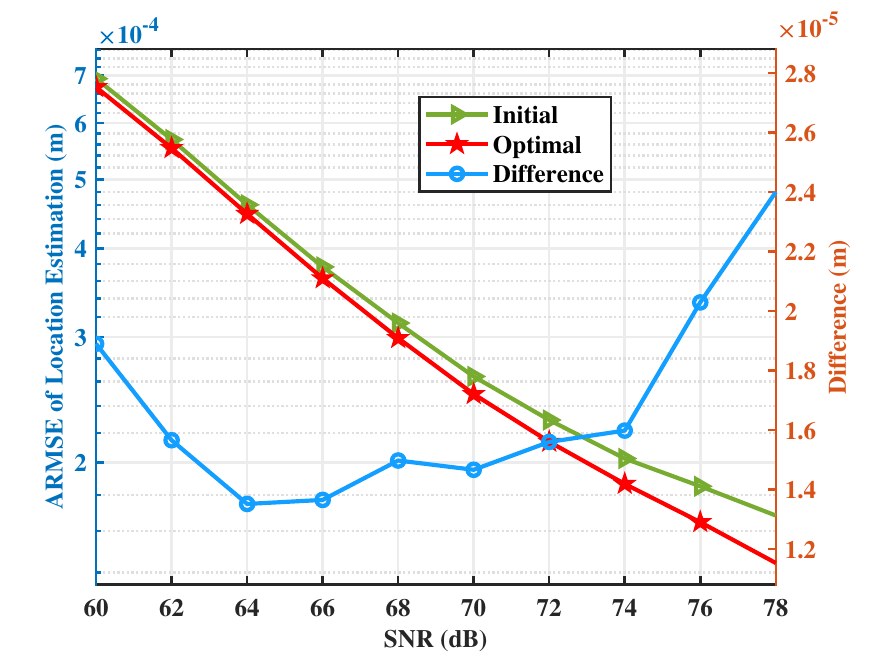}}
	\subfigure[Absolute velocity estimation] {\label{fig10.b}\includegraphics[width=0.33\textwidth]{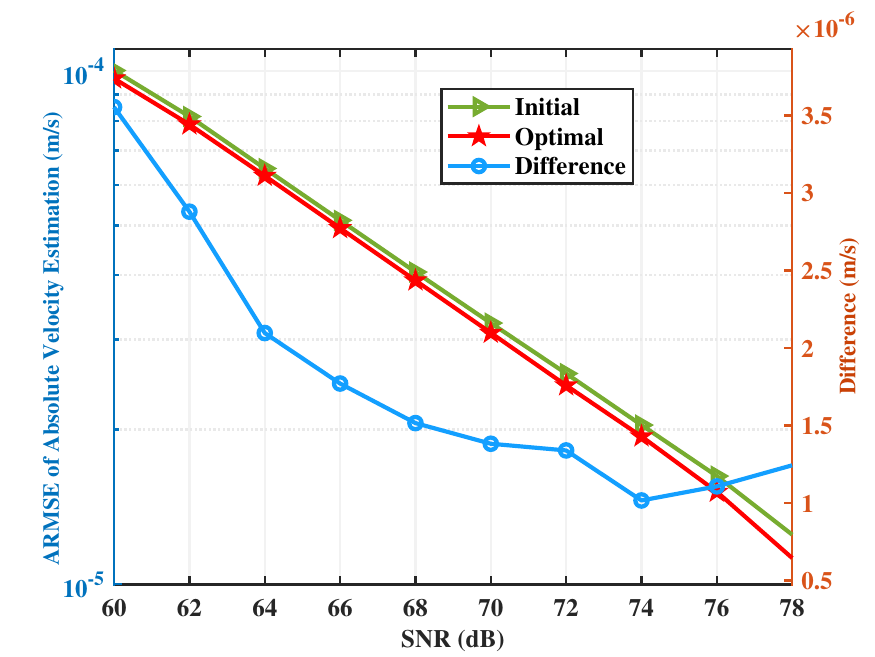}}
	\caption{ARMSEs under initial and optimized parameters, with \textit{Difference} denotes the value obtained by subtracting the ARMSE of \textit{Optimal} from the ARMSE of \textit{Initial}}
	\label{fig10}
\end{figure}

\section{Conclusion}\label{se6}
In this paper, we investigated the cooperative sensing of the CF-mMIMO ISAC system in 
multi-target scenarios and proposed a two-phase cooperative sensing framework to 
obtain the high-accuracy sensing performance. 
In the first phase, to address the challenge of joint optimization, 
the placement and antenna allocation of distributed APs were jointly optimized based on methods such as ADMM to lower the sensing CRBs. 
Based on the optimized parameters, 
an SL-MDTS scheme was proposed in the second phase to address the challenge of sensing information fusion, 
which combines symbol-level fusion with the optimization algorithm to 
obtain high-accuracy sensing performance. 
The simulation results validate that, compared to the state-of-the-art 
grid-based symbol-level sensing information fusion schemes, 
the proposed SL-MDTS scheme improves the accuracy of the location and velocity estimation
by 44\% and 41.4\%, respectively.
This work provides a reference for cooperative sensing in CF-mMIMO ISAC system.

\begin{appendices}
\section{Proof of Theorem 1} \label{apA}
In the CF-mMIMO ISAC system, we assume that 
the channel between each AP and the $u$-th target is independent and identically distributed (i.i.d.). 
According to (\ref{eq1}), the ISAC echo signal of the $l$-th AP canceling pilot data can be expressed as
\begin{equation}\label{eqap1}
    y_l^u = \sum_{p = 1}^{N_\text{A}^l}\sum_{n=1}^{N_\text{c}}\sum_{m=1}^{M}\left[\begin{array}{l}
      \alpha_l^u e^{j2\pi f_{\text{D},l}^u m T} e ^{-j2\pi n \Delta f \tau_l^u} \\ \times e ^{j2\pi p (\frac{d_\text{r}}{\lambda})\sin(\theta_l^u)} 
    \end{array}\right] + z_l^u,
\end{equation}
where $z_l^u \sim\mathcal{CN}\left(0,\sigma_\text{z}^2\right)$ is defined without loss of generality. 
Based on (\ref{eqap1}), the probability density function (PDF) of $\mathbf{t}_u^\text{N}$ is 
\begin{equation}\label{eqap2}
\begin{aligned}
&f\left(\mathbf{y};\mathbf{t}_u^\text{N}\right)=\frac{1}{\left(2\pi \sigma_\text{z}^2\right)^{L/2}} \\
& \quad \times \exp\left(\frac{-1}{2\sigma_\text{z}^2}\sum_{l=1}^L\left|y_l^u-\sum_p\sum_n\sum_m s_{l,p,n,m}^u\right|^2\right),
\end{aligned}  
\end{equation}
where $\mathbf{y} = \left[y_1^u,y_2^u,\cdots,y_L^u\right]^\text{T}$ denotes a set of echo signals received by $L$ APs, 
and 
\begin{equation}\label{eqap3}
 s_{l,p,n,m}^u= \alpha_l^u e^{j2\pi f_{\text{D},l}^u m T} e ^{-j2\pi n \Delta f \tau_l^u} e ^{j2\pi p (\frac{d_\text{r}}{\lambda})\sin(\theta_l^u)}.   
\end{equation}

The log-likelihood function of $\mathbf{t}_u^\text{N}$ is 
\begin{equation}\label{eqap4}
\begin{aligned}
\ln f\left(\mathbf{y};\mathbf{t}_u^\text{N}\right)=&\frac{-L}{2}\ln \left(2\pi \sigma_\text{z}^2\right)\\ & - \frac{1}{2\sigma_\text{z}^2} \sum_{l=1}^L\left|y_l^u-\sum_p\sum_n\sum_m s_{l,p,n,m}^u\right|^2 .   
\end{aligned}
\end{equation}
Based on (\ref{eqap4}) and the definition of FIM, 
the FIM of $\mathbf{t}_u^\text{N}$ is
\begin{equation}\label{eqap5}
 \mathbf{F}_\text{t} = \left[\begin{array}{ccc}
\mathbf{A}  & \mathbf{B} & \mathbf{C} \\
   \mathbf{B}  & \mathbf{D} & \mathbf{E} \\
   \mathbf{C}  & \mathbf{E} & \mathbf{F}
 \end{array}\right]\in \mathbb{R}^{3L\times 3L},
\end{equation}
where $\mathbf{A}$, $\mathbf{B}$, $\mathbf{C}$, $\mathbf{D}$, $\mathbf{E}$, 
and $\mathbf{F}$ are all the $L \times L$ diagonal matrices whose diagonal elements are expressed in \eqref{eq7} - \eqref{eq12}, respectively.

For the derivation of the Jacobian matrix, 
the unknown estimate vector $\mathbf{t}_u$ and transform vector $\mathbf{t}_u^\text{N}$ 
have the following geometric relationship
\begin{equation}\label{eqap6}
    \begin{cases}
        \tau_l^u = \frac{2\|\mathbf{c}_l-\mathbf{t}_u\|_2}{c} \\
        y_u^\text{tr} - y_l = \tan\left(\theta_l^u\right)\left(x_u^\text{tr} - x_l\right) \\
        f_{\text{D},l}^u = \frac{-2f_\text{c}v_u\cos\left(\theta_l^u-\theta_u\right)}{c}
    \end{cases}.
\end{equation}
Therefore, combining (\ref{eq5}) and (\ref{eqap6}), 
we can obtain the expression of $\mathbf{P}$ as shown in (\ref{eq13}).

\section{Solution of Problem \hyperref[eq23]{(P2.1)}}\label{apex2}
Considering the associated term $l$ of the unconstrained convex optimization problem \hyperref[eq23]{(P2.1)}, 
this is a minimization problem with a linear term and a quadratic penalty term. 
According to the first-order definition of convex optimization~\cite{nocedal1999numerical}, 
the gradient at the global optimal value is zero.
Therefore, we can calculate the gradient of the objective function with respect to $\mathbf{a}_l^k$, i.e., 
\begin{equation}\label{eqap7}
    \nabla_{\mathbf{a}_l^k} = \boldsymbol{\lambda}_l^k + \rho_1\left(\mathbf{a}_l^k - \mathbf{C}_l^\text{t}\mathbf{z}^k +  \hat{\mathbf{c}}_l \right),
\end{equation}
If $\nabla_{\mathbf{a}_l^k} = 0$, 
the global optimal value is $\mathbf{a}_l^* = \mathbf{C}_l^\text{t}\mathbf{z}^k -  \hat{\mathbf{c}}_l - \frac{\boldsymbol{\lambda}_l^k}{\rho_1}$. 
To satisfy the constraint (\ref{eq23a}), 
the projection operation need to performed to ensure that the optimal value falls within the feasible domain. 
The final updated value is 
\begin{equation}
    \mathbf{a}_l^{k+1} = \left\{\begin{matrix}
 \mathbf{C}_l^\text{t}\mathbf{z}^k -  \hat{\mathbf{c}}_l - \frac{\boldsymbol{\lambda}_l^k}{\rho_1}, \quad \|\mathbf{C}_l^\text{t}\mathbf{z}^k -  \hat{\mathbf{c}}_l - \frac{\boldsymbol{\lambda}_l^k}{\rho_1}\|_2 \le \varepsilon  \\
 \frac{\mathbf{C}_l^\text{t}\mathbf{z}^k -  \hat{\mathbf{c}}_l - \frac{\boldsymbol{\lambda}_l^k}{\rho_1}}{\|\mathbf{C}_l^\text{t}\mathbf{z}^k -  \hat{\mathbf{c}}_l - \frac{\boldsymbol{\lambda}_l^k}{\rho_1}\|_2} \varepsilon, \quad \|\mathbf{C}_l^\text{t}\mathbf{z}^k -  \hat{\mathbf{c}}_l - \frac{\boldsymbol{\lambda}_l^k}{\rho_1}\|_2 \ge \varepsilon
\end{matrix}\right.,
\end{equation}
which can be simplified as (\ref{eq24}).

\end{appendices}

\bibliographystyle{IEEEtran}
\bibliography{reference}

\end{document}